\documentclass[reqno,11pt]{amsart}
\usepackage{amsmath}
\usepackage{latexsym,amsmath}
\usepackage{amsthm}
\usepackage{amssymb}
\usepackage{epsfig, verbatim}
\usepackage{version}
\usepackage{graphicx}
\usepackage{color}

\newtheorem{theorem}{Theorem}
\newtheorem{lemma}[theorem]{Lemma}

\newtheorem{corollary}[theorem]{Corollary}

\newtheorem{claim}{Claim}

\def\etal{{\em et. al.}~}
\def\ie{{\em i.e. }}

\def\M{M{\footnotesize{ULTICUT}}}
\def\VM{V{\footnotesize{ERTEX}} M{\footnotesize{ULTICUT}}}
\def\CM{C{\footnotesize{OMPONENT}} M{\footnotesize{ULTICUT}}}
\def\BM{B{\footnotesize{ACKBONE}} M{\footnotesize{ULTICUT}}}

\def\MW{M{\footnotesize{ULTIWAY}} C{\footnotesize{UT}}}
\def\VC{V{\footnotesize{ERTEX}} C{\footnotesize{OVER}}}

\def\MIT{M{\footnotesize{ULTICUT}} I{\footnotesize{N}} T{\footnotesize{REES}}}

\begin{document}

\title[Multicut is FPT]
{Multicut is FPT}
\author{Nicolas Bousquet}
\address{Universit\'e Montpellier 2 - CNRS, LIRMM, 161 rue Ada, 34392 Montpellier, France}
\email{bousquet,daligault,thomasse@lirmm.fr}
\author{Jean Daligault}
\author{St\'ephan Thomass\'e}

\begin{abstract}
Let $G=(V,E)$ be a graph on $n$ vertices and $R$ 
be a set of pairs of vertices in $V$
called \emph{requests}. A \emph{multicut}
is a subset $F$ of $E$ such that every 
request $xy$ of $R$ is cut by $F$, \ie{}every $xy$-path of $G$ intersects $F$. 
We show that there exists an $O(f(k)n^c)$ algorithm 
which decides if there exists a 
multicut of size at most $k$. In other words, the \M{} problem parameterized by the solution size $k$ is Fixed-Parameter Tractable. The proof extends to vertex-multicuts.
\end{abstract}

%\bigskip

\maketitle

\section{Introduction}
The study of cuts and flows is one of the most active field in combinatorial
optimization. However, if the simplest case involving one source and one sink is algorithmically
tractable, the problem becomes hard as soon as one deals with multiple terminals.
For instance, given two requests $(x_1,y_1)$ and $(x_2,y_2)$ in a directed
graph $D$, it is NP-complete to decide if there exist two disjoint directed
paths, respectively from $x_1$ to $y_1$ and from $x_2$ to $y_2$ \cite{2dirNPhard}. In fact, even deciding
if two given vertices belong to a directed circuit is already hard.
The picture changes when considering undirected graphs, in which case the celebrated
result of Robertson and Seymour \cite{GM13} asserts that given $k$ requests
$(x_1,y_1),(x_2,y_2),\dots ,(x_k,y_k)$, one can decide in cubic time if there
exists $k$ disjoint paths connecting all pairs $x_i,y_i$. The catch is
of course that cubic time refers to the instance size, which we generally
denote by $n$. From their work, the complexity of the $k$-path problem is
$O(f(k)n^3)$, where $f$ is by no mean polynomial, since
the question is NP-complete when $k$ is part of the input.

This result received considerable attention, both since this is the key tool
for computing a given minor in a graph, but also because it has opened a breach
in the classical NP-complete/P duality. Indeed, the difficulty of the $k$-path problem
does not depend on the size of the instance, but rather on the number of paths we
are looking for. In other words, the parameter containing the hardness of the
problem is the number $k$ of paths. In a more general way, a
problem is \emph{fixed parameter tractable} (FPT) with respect to the parameter
$k$ (e.g. solution size, treewidth, ...) if for any instance of size $n$ it can be
solved in time $O(f(k)n^c)$ for some fixed $c$. The reader is invited to refer to
now classical books by Downey and Fellows \cite{DF99}, Flum and Grohe \cite{FG06} and Niedermeier \cite{Nie06}.

The dual problem of finding disjoint paths from a source $s$ to a sink $t$
is the cut problem where one asks for a set of vertices
or edges which deletion separates $s$ from $t$. Menger's theorem, or more generally
LP-duality, asserts that the maximum number of disjoint $st$-paths is equal to the
minimum size of a cut. This property no longer holds when considering multiple
requests, where the maximum number of disjoint paths connecting requests is
only an obvious lower bound for the size of a multicut, \ie{}a set of vertices
or edges which deletion separates $x_i$ from $y_i$ for every
request $x_iy_i$. Formally, we have the following problem:

\begin{quote}
\noindent
\M{}:\\
\textbf{Input}: A graph $G$, a set of requests $R$, an integer $k$.\\
\textbf{Parameter}: $k$.\\
\textbf{Output}: TRUE if there is a multicut of size at most $k$, otherwise FALSE.
\end{quote}

The status of this problem is one of the long standing open problems in parameterized complexity.
The main result of this paper is to provide an FPT algorithm for \M{}.
The proofs being slightly less cumbersome in the edge case, we present our work
in terms of edge-multicut rather than vertex-multicut. Our last section
shows how to translate the algorithm to deal with vertices. Let us now give some formal definitions.
Given a graph $G$ and a set $R$ of pairs of distinct vertices called \emph{requests},
a \emph{multicut} is a subset $F$ of edges of $G$ such that every $xy$-path, where
$xy$ is a request, contains an edge of $F$. Equivalently, the two endpoints (or \emph{terminals})
of every request of $R$ belong to different connected components of $G\setminus F$.

The \M{} problem is already hard on trees since \VC{} is equivalent to \M{} in stars.
Hence \M{} is NP-complete and Max-SNP hard. \M{} and its variants have raised
an extensive literature.  These problems play an important role in network
issues, such as routing and telecommunication (see \cite{CLR05}).

\MIT{} was already a challenging problem. Garg \etal \cite{GVY97} proved that it admits a
factor 2 approximation algorithm. Guo and Niedermeier \cite{GN05} proved that
\MIT{} is FPT with respect to the solution size. Bousquet \etal
\cite{BDTY09} provided a polynomial kernel.

Another variant is the \MW{} problem in
which  a set of (non-paired) terminals has to be pairwise separated. When parameterized
by the solution size, \MW{} has been proved FPT by Marx \cite{marxmultiwaycut}. A
faster $O^\star(4^k)$ algorithm is due to Chen \etal \cite{chenmultiwaycut}.

On general instances, Garg \etal gave an approximation algorithm for \M{} within a logarithmic factor in \cite{GVY93}.
However \M{} has no constant factor approximation algorithm if Khot's Unique Games Conjecture holds \cite{MCnonapprox}. This fact is a further motivation to the study of the fixed parameterized tractability of \M{}.
Guo \etal showed in \cite{GHKNU06} that \M{} is FPT when parameterized by both the treewidth of the graph and the number of requests. Gottlob and Lee in \cite{GL07} proved a stronger result: \M{} is FPT when parameterized by the treewidth of the input structure, namely the input graph whose edge set is completed by the set of request pairs.
Recently, Daligault \etal~\cite{DPPT09} proved that \M{} can be reduced to instances in which the
graph $G$ has treewidth bounded in terms of $k$.

The graph minor theorem of Roberston and Seymour implies that \M{} is non-uniformly FPT when
parameterized by the solution size and the number of requests. Marx proved that \M{} is (uniformly)
FPT for this latter parameterization \cite{marxmultiwaycut}. A faster algorithm running in time
$O^\star(8 \cdot l)^k$ was given by Guillemot \cite{GuillemotIWPEC}. Marx \etal \cite{MarxTwReduc}
obtained FPT results for more general types of constrained \M{} problems through treewidth reduction
results. However their treewidth reduction techniques do not yield the tractability of \M{} when parameterized
only by the solution size. Recently, Marx and Razgon obtained a factor 2 Fixed-Parameter-Approximation for \M{}
in \cite{MulticutFPTapprox}.

Marx and Razgon independently found a proof of the fact that \M{} is FPT, with a
rather different approach, see~\cite{MR10}.

The outline of our proof is quite pedestrian, even if some of the partial results
are still a bit technical. Informally, our goal is to reduce our input graph
to a subdivision of a graph with a bounded number of edges. The crucial
tool for this reduction is to find a partition of the set of all "reasonable" solutions
of the \M{} problem into a bounded number of subsets, in which the
multicuts are "totally ordered" (this will be formalized in Subsection~\ref{dilworth}). This application of Dilworth's Theorem is maybe
the most interesting aspect of our proof. However, it requires an important cleaning
of our instance before being applied. It would be interesting to find an adequate
partial order on multicuts at an earlier stage of our proof to shorten the argument
and get more insight in the structure of multicuts.

In Section~\ref{sectionconnectivity} we
develop some connectivity tools which are used to certify
that some requests are irrelevant. In Section~\ref{CM}, we first compute a vertex-multicut
$Y$ of size $k+1$ by iterative compression. We then reduce to the case where each
component of $G\setminus Y$ has one or two attachment vertices in $Y$. We show that
components with one attachment vertex have only a bounded (in $k$) number of
terminals. For components with two attachment vertices, we identify
one path (called \emph{backbone}) where only one edge is chosen in the multicut.
In Section~\ref{trlemon}, we reduce to the case where the backbone is the only path connecting
the two attachment vertices. Finally, we show that \M{} can be reduced to an instance
which is a subdivision of a graph with a bounded number of edges and where the multicut
selects a single edge in each of the paths corresponding to the subdivided edges.
This case can easily  be coded by a 2-SAT instance, which is solvable in polynomial time.
In the Appendix (Section~\ref{time}) we improve the running time of our algorithm to a single
exponential in terms of $k$, and in Section~\ref{vertex} we sketch how to adapt our proof
to the case of vertex-multicuts.

\emph{Acknowledgments.} This proof was a long process in which several researchers gave us advice or
participated in some of the partial results. One of our important tools, Lemma~\ref{extractright},
is the crucial result of the reduction of \M{} to bounded treewidth, and was
obtained in collaboration with Christophe Paul and Anthony Perez \cite{DPPT09}. The fact that
a vertex cutset $Y$ could be obtained by iterative compression, a tool
of \cite{MulticutFPTapprox}, was brought to our attention by Sylvain Guillemot. Finally,
we would like to thank Valentin Garnero who worked on this project as a part
of his graduate research experience in June 2010.

%%%%%%%%%%%%%%%%%%%%%%%%%%%%%%%%%%%%%%%%%%%%%%%%%%%%%%%%%%%%%%
%%%%%%%%%%%%%%%%%%%%%%%%%%%%%%%%%%%%%%%%%%%%%%%%%%%%%%%%%%%%%%
%%%%%%%%%%%%%%%%%%%%%%%%%%%%%%%%%%%%%%%%%%%%%%%%%%%%%%%%%%%%%%
%%%%%%%%%%%%%%%%%%%%%%%%%%%%%%%%%%%%%%%%%%%%%%%%%%%%%%%%%%%%%%
%%%%%%%%%%%%%%%%%%%%%%%%%%%%%%%%%%%%%%%%%%%%%%%%%%%%%%%%%%%%%%
%%%%%%%%%%%%%%%%%%%%%%%%%%%%%%%%%%%%%%%%%%%%%%%%%%%%%%%%%%%%%%
%%%%%%%%%%%%%%%%%%%%%%%%%%%%%%%%%%%%%%%%%%%%%%%%%%%%%%%%%%%%%%
%%%%%%%%%%%%%%%%%%%%%%%%%%%%%%%%%%%%%%%%%%%%%%%%%%%%%%%%%%%%%%

\section{Preliminaries}

A vertex which sends a request is called a \emph{terminal}.
The number of requests sent by a terminal is its 
\emph{request degree}.
We study \M{} variants with additional contraints on the deleted edges. 
In the original \M{} problem, we can delete any set of $k$ edges, but in some more 
constrained versions we must delete a prescribed number of edges on some particular paths.
The total number of deleted edges is called \emph{deletion allowance}
of the multicut problem. We will make extensive use of the 
term \emph{bounded} which always implicitely means bounded in 
terms of the deletion allowance. Also, when speaking of 
FPT time, we always mean $O(f(d)n^c)$ where $c$
is a fixed constant and $d$ is the deletion allowance. Let us 
discuss further some of the operations we will often perform.

\emph{Reductions.} These are computations where the
output is a new instance which is equivalent to the original instance with respect
to the existence of a solution. One of the most natural reductions 
concerns irrelevant requests, i.e. a request $xy$ such that 
every multicut of $R\setminus xy$ actually cuts
$x$ from $y$, where $R$ is the set of requests. If one can certify that a request $xy$ is irrelevant,
the reduction just consists in replacing $R$ by $R\setminus xy$.
The difficulty is obviously to certify that a request is irrelevant.
Another easy reduction is obtained if we can certify that there
exists a multicut which does not separate 
two given vertices $u$ and $v$. In this case we simply contract $u$
and $v$. Reductions are easy to control since
we can perform them freely provided that some invariant polynomial 
in $n$ decreases. For instance, request deletions can be performed
at most $n^2$ times, and vertex contractions at most $n$ times.

\emph{Branchings.} In our algorithm, we often have
to decide if the multicut we are looking for is of a particular
type, where the number of types is bounded. We will then say that we branch over 
all the possible cases. This means that to compute the result
of the current instance, we run our algorithm on each case,
in which we force the solution to be of each given type. The output is TRUE if 
at least one of the outputs returns TRUE. 
To illustrate this, in the case of a graph $G$ with two 
connected components $G_1$ and $G_2$, both containing requests, we would 
branch over $k-1$ instances, depending of the number of edges (between 1 and $k-1$) that we delete 
from $G_1$. This simple branching explains why we can focus on connected
graphs.

\emph{Invariants.} To prove that the number of branchings is bounded, 
we show that some invariant is modified at each 
branching step, and that the number of times that this invariant can be
modified is bounded. We usually have several invariants
ordered lexicographically. In other words, we have different
invariants which we want to increase or decrease and can take 
a bounded number of values. These invariants are ordered, there 
is a primary invariant, a secondary, etc. Each branching
must \emph{improve} our invariant, \ie{}the first invariant
(with respect to priority order) which is changed by the branching
must be modified according to the preference, increase or decrease,
that we specified for it. For instance the primary invariant could be the 
number of deleted edges, which we want to decrease and the secondary 
invariant could be the connectivity of $G$, which we want to increase.
If we can decrease the number of deleted edges we do so even if the 
connectivity of the graph decreases. Also, if a branching increases
connectivity and leaves the number of deleted edges unchanged, we improve 
our invariant.

\section{Connectivity in FPT time.}\label{sectionconnectivity}

Dealing with minimum cuts can be done in polynomial time 
with usual flow techniques. However, dealing with $k$-edge cuts 
when $k$ is some fixed value larger than the optimum is 
more difficult. We develop here some tools to deal in FPT
time with bounded cuts, of the same flavour as in \cite{DFVS}, \cite{DPPT09} and \cite{marxmultiwaycut}.

\subsection{Enumerating cuts in FPT time}
Let $G$ be a connected graph on $n$ vertices with a particular vertex $x$ called 
\emph{root}. We deal in this part with cuts, \ie{}bipartitions of the vertex set of $G$. To fix one side
of our cuts, we define a \emph{cut} as a subset of vertices $S$ containing
$x$. The \emph{border} of $S$ is the set of edges of $G$
with exactly one endpoint in $S$. We denote it by 
$\Delta (S)$. Its cardinality is denoted by $\delta (S)$.
Recall that the function $\delta $ is submodular, 
\ie{}$\delta (A)+\delta (B)\geq \delta (A\cap B)+\delta (A\cup B)$.
Given a subset $Y$ of vertices of $V(G)$, we denote its \emph{complement}
$V(G)\setminus Y$ by $\overline Y$.
A cut $S$ is a \emph{left cut} if every cut $T\subsetneq S$ satisfies 
$\delta (T)>\delta (S)$. Note that every cut $S$ contains a left
cut $S'$ with $\delta (S')\leq \delta (S)$.

\begin{lemma} \label{unioncut}
Left cuts are closed under union.
\end{lemma}

\begin{proof} 
Let $S_1\cup S_2$ be the union of two left cuts. Let $S_3\subsetneq S_1\cup S_2$
be a cut with minimum border. Without loss of generality, we assume that 
$S_1$ is not included in $S_3$. Since $S_1$ is a left cut, $\delta (S_1\cap S_3)>\delta (S_1)$.
Since $\delta (S_1\cap S_3)+\delta (S_1\cup S_3)\leq \delta (S_1)+\delta (S_3)$,
we obtain $\delta (S_1\cup S_3)< \delta (S_3)$. Furthermore 
$S_1\cup S_3\subseteq S_1\cup S_2$. Since $S_3$ has minimum border
among strict subsets of $S_1\cup S_2$, we have 
$S_1\cup S_3=S_1\cup S_2$. Finally $\delta (S_1\cup S_2)< \delta (S_3)$, thus
$S_1\cup S_2$ is a left cut.
\end{proof}

\begin{lemma} \label{smallsizecut}
If $S_1,S_2$ are distinct left cuts, $\delta (S_1\cup S_2)<\max (\delta (S_1),\delta (S_2))$.
\end{lemma}
 
\begin{proof}
Since $S_1\neq S_1\cup S_2$ or $S_2\neq S_1\cup S_2$, we assume without 
loss of generality $S_1\subsetneq S_1\cup S_2$.
By Lemma~\ref{unioncut}, $\delta (S_1\cup S_2)<\delta (S_1)\leq \max (\delta (S_1),\delta (S_2))$.
\end{proof}

A cut $S$ is \emph{indivisible} if $G\setminus S$ is connected,
otherwise 
it is \emph{divisible}.

\begin{lemma} \label{divind}
If $S$ is a divisible left cut and $Y$ is a connected component
of $G\setminus S$, the cut $\overline Y$ is an indivisible left cut
with $\delta (\overline Y)< \delta (S)$.
\end{lemma}

\begin{proof}
The cut $\overline Y$ is indivisible by construction and 
$\Delta (\overline Y)\subsetneq \Delta(S)$, so we just have to prove that $\overline Y$ 
is a left cut. Consider a left cut $T\subseteq \overline Y$ which minimizes
$\delta (T)$. By Lemma~\ref{unioncut}, $S\cup T$ is a left cut. Moreover 
$\delta (T)\leq \delta (S\cup T)$ by minimality of $\delta(T)$, hence 
$T=S\cup T$. In particular $S\subseteq T$. Since every edge of 
$\Delta (\overline Y)$ joins $S$ to $Y$, we have $\Delta (\overline Y)\subseteq \Delta (T)$. Therefore by
minimality of $\delta(T)$, we have $T=\overline Y$. Thus $\overline Y$ is a left cut.
\end{proof}

\begin{corollary} \label{indlef}
Every indivisible cut $S$ contains an indivisible left cut $S'$
with $\delta (S')\leq \delta (S)$.
\end{corollary}

\begin{proof}
Let $S''$ be a left cut contained in $S$ such that $\delta (S'')\leq \delta (S)$. 
We assume that $S''$ is divisible, otherwise $S':=S''$.
Let $Y$ be the component of $G\setminus S''$ which contains $\overline S$. 
By Lemma~\ref{divind}, $S':=\overline Y$ is an indivisible left cut with 
$\delta(S')< \delta (S'') \leq \delta (S)$. Moreover $S'\subseteq S$ since $\overline S\subseteq Y$.
\end{proof}

Given a vertex $y$, an \emph{$xy$-cut} is a cut 
$S$ such that $y\notin S$. We denote by $C^y_{k}$ the set of indivisible left $xy$-cuts with border $k$. 
We also denote by $C^y_{<k}$ (resp. $C^y_{\leq k}$) the union of $C^y_{i}$ for $i<k$
(resp. for $i\leq k$). 

\begin{theorem} \label{boundsize}
The set $C^y_{\leq k}$ has size at most 
$k!$ and can be computed in FPT time.
\end{theorem}

\begin{proof}
We prove this result by induction on $k$. Let us start 
our induction with $k=\lambda$, the first value such that $C^y_{\leq \lambda}$ is 
non empty. In other words, $\lambda$ is the edge-connectivity between $x$ and $y$.

\begin{claim} \label{init}
The set $C^y_{\leq \lambda}$ has only one element $S_{\lambda}$. It
can be computed in polynomial time. Moreover 
$S_{\lambda}$ contains all left
$xy$-cuts.
\end{claim}

\begin{proof}
The $\delta$ function is modular on $xy$-cuts of size $\lambda$. In particular, 
$C^y_{\leq \lambda}$ is closed under intersection, hence $S_{\lambda}$
is just the intersection of all $xy$-cuts with border $\lambda$.
It is indivisible by minimality of $\lambda$ and can be computed in polynomial time.
Now, if $T$ is a left $xy$-cut, we have 
$\delta(T\cap S_{\lambda})+\delta(T\cup S_{\lambda})\leq \delta(T)+\lambda$. By 
minimality of $\lambda$, $\delta(T\cup S_{\lambda})\geq \lambda$. Thus $\delta(T\cap S_{\lambda})\leq \delta(T)$.
The set $T$ being a left cut, we have $T\cap S_{\lambda}=T$, thus $T\subseteq S_{\lambda}$.
\end{proof}

Let $S\in C^y_k$ with $k>\lambda$. Let $T\in C^y_{<k}$, minimal with respect 
to inclusion which contains $S$. Such a cut $T$ exists since by Claim~\ref{init},
the cut $S_{\lambda}$ contains $S$. Since $S$ is indivisible, there exists 
an edge $e$ in $\Delta (T)\setminus \Delta (S)$. 

\begin{claim} \label{unicity}
If $S'\in C^y_{\leq k}$ is included in $T$ and $e\notin \Delta (S')$, then $S'=S$.
\end{claim}

\begin{proof}
Assume for contradiction that $S'$ is different from $S$. By Lemma~\ref{unioncut}, $S\cup S'$ is a left
$xy$-cut, and by Lemma~\ref{smallsizecut}, we have $\delta (S\cup S')<k$. 
Let $Y$ be the component of 
$G\setminus (S\cup S')$ which contains $\overline T$. By Lemma~\ref{divind},
the cut $\overline Y$ belongs to $C^y_{<k}$. 
Therefore $\overline {Y}\subseteq T$, and by minimality
of $T$, we have $\overline {Y}=T$. But this would mean $e\in \Delta (S\cup S')$, which
is a contradiction.
\end{proof}

Now we turn Claim~\ref{unicity} into an algorithm. 
For every cut $T$ in $C^y_{<k}$ and every edge $e\in \Delta (T)$, we contract $(G\setminus T)\cup e$
to a single vertex still called $y$. We call this graph $G'$. If the $xy$-edge connectivity of $G'$ is not equal
to $k$, the search stops. Otherwise, we compute the unique
indivisible left $xy$-cut $S$ with border $k$. This cut 
$S$ in $G$ is an element of $C^y_k$. By
Claim~\ref{unicity}, all the elements of $C^y_k$ can be computed 
in this way. This algorithm gives the upper bound $|C^y_{\leq k}|\leq |C^y_{<k}|+(k-1)|C^y_{<k}|$, 
hence $|C^y_{\leq k}|\leq k!$. This concludes the proof of Theorem~\ref{boundsize}.
\end{proof}

The value $k!$ in Theorem~\ref{boundsize} can actually be improved to $4^k$ \cite{MarxSlides}.

\subsection{Irrelevant requests}

We denote by $C_k$ the union of $C^y_k$ for all $y$, by $C_{<k}$ the 
union of $C^y_{<k}$ for all vertices $y$ of $G$, and by $C_{\leq k}$ the 
union of $C^y_{\leq k}$ for all vertices $y$ of $G$. 
A collection of sets is called a \emph{$\Delta$-system} if every two distinct sets 
have the same intersection. Erd\H{o}s and Rado~\cite{ER} proved that there exists 
a function $er$ such that any collection of $er(k,r)$ sets with size at most $k$
contains a $\Delta$-system consisting of $r$ sets. The bound in the next result
will be improved in Section~\ref{time}.

\begin{theorem} \label{extractleft}
Every set $K$ with at least $er(k!,k')$ vertices of $G$ contains a subset 
$K'$ of size $k'$ such that every left cut $S$ with
$\delta(S)\leq k$ satisfies either $S\cap K'=\emptyset$ or $|K'\setminus S|\leq k$.
The set $K'$ can be computed in FPT time.
\end{theorem}

\begin{proof}
In other words, every left cut with border at most $k$ isolates either all the elements
of $K'$, or at most $k$ elements of $K'$.
Let us consider the collection $\mathcal{C}$ of sets $C^y_{\leq k}$, for all $y\in K$.
By Theorem~\ref{boundsize}, the collection $\mathcal{C}$ has size
bounded in terms of $k$ and $k'$ and can be computed 
in FPT time.
Since the sets $C^y_{\leq k}$ have size at most $k!$ and 
the set $K$ has size $er(k!,k')$, there exists a
$\Delta $-system of size $k'$, \ie{}a subset 
$K'$ of $k'$ vertices of $K$ such that for all $y,y'\in K'$, we have 
$C^{y'}_{\leq k}\cap C^{y}_{\leq k}$ equal to some fixed set $C$
of $C_{\leq k}$. This set $K'$ is computable in bounded time.
Every cut $S$ in $C$ satisfies $S\cap K'=\emptyset$, \ie{}the cuts in $C$ isolate $K'$. Moreover,
if a cut $S$ in $C_{\leq k}$ does not belong to $C$, then
$S$ belongs to at most one $C^y_{\leq k}$, hence isolates at most
one vertex of $K'$.

Thus we have proved so far that the conclusion of Theorem~\ref{extractleft}
holds if $S$ is an indivisible left cut with border or size at most $k$, with 
the stronger conclusion that $S$ isolates at most one vertex of 
$K'$ when it does not completely cut $K'$. To obtain our conclusion,
let us observe that if $S$ is divisible and $Z$ is a component of 
$G\setminus S$, then by Lemma~\ref{divind} the cut $\overline Z$
belongs to $C_{\leq k}$. Hence either $\overline Z$ isolates $K'$, or 
$\overline Z$ isolates at most one vertex of $K'$. Since the number of components of $G\setminus S$
is at most $k$, we have our conclusion.
\end{proof}

The next result is our key tool in reducing \M{} to graphs of bounded 
treewidth in~\cite{DPPT09}, but the following easy argument is more 
straightforward. The proof of Theorem 4 of~\cite{DPPT09} implies that the 
following result holds with $h(\ell)=\ell^{O(\ell)}$, and can be computed 
in time $\ell^{O(\ell)}n^c$.
 
\begin{theorem} \label{extractright}
Every set $K$ with at least $h(\ell):=\ell.2^{\ell !}+1$ 
vertices of $G$ contains a vertex $y$ such that every cut $S$ with $\delta (S)+|S\cap K|\leq \ell$
is such that $y\notin S$. Moreover, $y$ is computed in FPT time.
\end{theorem}

\begin{proof}
In other words, whenever the deletion of a set of $a$ edges 
isolates $x$ from all but $b$ elements of $K$, with $a+b\leq \ell$, then
the vertex $y$ is also isolated from $x$.
We create a new vertex $z$ joined to all the vertices of $K$ to form 
a new graph $G'$. In this proof we change our point of view and 
consider that $z$ is our root vertex for $G'$, \ie{}$z$ plays
for $G'$ the same role as $x$ plays for $G$. In $G'$, the set 
$C$ of indivisible left $zx$-cuts with border 
at most $\ell$ has size at most $\ell !$ by Theorem~\ref{boundsize}. Hence, since the size of $K$
is at least $\ell.2^{\ell !}+1$, there exists a subset $T$ of $K$
of size at least $\ell +1$ such that for every cut $S$
in $C$, we have either $T\subseteq S$ or $T\cap S=\emptyset$.
We compute such a set $T$ in FPT time. We then pick $y$ in $T$.
Let us prove that $y$ satisfies 
the conclusion of Theorem~\ref{extractright}.

In the graph $G$, consider a set $A$ of $a$ edges which isolates $x$ from all 
the elements of $K$ save a subset $B$ of size $b$ with $a+b\leq \ell$.
Let $F$ be the set of edges $A\cup \{zb~:~b\in B\}$ of $G'$. Note
that $F$ is a $zx$-edge cut. We denote by $X$ the component of 
$x$ in $G'\setminus F$. Since $V(G')\setminus X$ is an indivisible $zx$-cut 
with border at most $\ell$, it contains by Corollary~\ref{indlef} a indivisible left $zx$-cut $S$
with border at most $\ell$, in other words $S$ belongs to $C$.

Let us first observe that the set $T$ cannot be disjoint from $S$. Indeed
$T$ has size $\ell +1$ and each of its elements is joined to $z$, thus the border
of $S$ would exceed $\ell$. Hence $T$ is included in $S$, thus
the set of edges $A$ isolates $T$ from $x$, and in particular $y$ from $x$. This concludes our proof.
\end{proof}

\begin{corollary} \label{boundeddegree}
In a multicut problem instance with deletion allowance $k$, the maximum request degree
can be reduced to at most $h(k+1)$ in FPT time.
\end{corollary}

\begin{proof}
Let $x$ be a vertex which sends requests to all vertices of $K$, where $|K|\geq h(k+1)$. By 
Theorem~\ref{extractright}, there is a vertex $y$ of $K$ such that 
every subset $S$ containing $x$ such that $\delta (S)+|S\cap K|\leq k+1$
is such that $y\notin S$. We simply delete the request $xy$ from 
the set of requests. Indeed, let $F$ be a multicut with 
size at most $k$ of this reduced instance. Let $S$ be the component 
of $x$ in $G\setminus F$. Since $F$ is a multicut, no element 
of $K\setminus y$ belongs to $S$. Moreover $\delta (S)\leq k$ since
at most  $k$ edges are deleted. Thus $\delta (S)+|S\cap K|$ is at most $k+1$, hence 
this implies that $y\notin S$. In other words, even if we do not
require to cut $x$ from $y$, any multicut of the reduced instance
will cut the request $xy$. Therefore the request $xy$ can be deleted from $R$.
\end{proof}

%%%%%%%%%%%%%%%%%%%%%%%%%%%%%%%%%%%%%%%%%%%%%%%%%%%%%%%%%%%%%%%%%%
%%%%%%%%%%%%%%%%%%%%%%%%%%%%%%%%%%%%%%%%%%%%%%%%%%%%%%%%%%%%%%%%%%+
%%%%%%%%%%%%%%%%%%%%%%%%%%%%%%%%%%%%%%%%%%%%%%%%%%%%%%%%%%%%%%%%%%
%%%%%%%%%%%%%%%%%%%%%%%%%%%%%%%%%%%%%%%%%%%%%%%%%%%%%%%%%%%%%%%%%%

\subsection{Cherry reduction}

An \emph{$x$-cherry}, or simply \emph{cherry} is a connected induced
subgraph $C$ of $G$ with a particular  vertex $x$ called \emph{attachment vertex} of $C$
such that there is no edge from $C\setminus x$ to $G\setminus C$ and 
no request has its two terminals in $C\setminus x$. In other words the only requests 
inside an $x$-cherry have origin $x$. Note that we can always assume that the restriction 
of a multicut to an $x$-cherry $C$ is the border of a left cut of $C$. 
If $u\in C\setminus x$, a request $uv\in R$ is \emph{irrelevant} if
every multicut $F$ with at most $k$ edges of the reduced instance on $R\setminus uv$
and such that $F\cap C$ is the border of a left cut in $C$ actually separates $u$ from $v$.

\begin{theorem} \label{cherryrequestreduction}
Let $C$ be an $x$-cherry of a instance with deletion allowance $k$. We can 
find in FPT time a set $K(C)$ of at most $b(k):=h(k+1).er(k!,h(2k+1))$ terminals in 
$C\setminus x$, such that if $F$ is a set of at most $k$ edges which cuts all 
requests with one endpoint in $K(C)$ and such that $F\cap C$ is the border of a left cut, then $F$ actually cuts all requests with an endpoint in $C\setminus x$. 
\end{theorem}

\begin{proof}
By Corollary \ref{boundeddegree}, we can assume that
all terminals have request degree at most $h(k+1)$. 
Let $L$ be the subset of terminals of $C\setminus x$. 
We assume $|L|> b(k)$. 
Our goal is to show that there exists an irrelevant 
request with one endpoint in $L$. Let us consider the bipartite graph
formed by the set of requests with one endpoint in $L$. It is bipartite since $C\setminus x$ has no internal requests. Recall that
if a bipartite graph $(X,Y)$ has maximum degree $d$ and minimum degree one, there exists 
a matching with at least $|X|/d$ edges. To see this, observe that the edges of $(X,Y)$ can be partitioned into $d$ matchings
and that the graph contains at least $|X|$ edges. 

The request graph thus contains a matching $M$
of size at least $er(k!,h(2k+1))$ such that each request in $M$ has 
one endpoint in $L$ and the other endpoint out of $C\setminus x$. Let 
$K:=V(M)\cap V(C\setminus x)$. We first only consider the cherry $C$ where $x$ is the root. 
Since the size of $K$ is at least $er(k!,h(2k+1))$,
the set $K$ contains by Theorem~\ref{extractleft} a subset $K'$ of size $h(2k+1)$ such that every left cut $S$ with
border at most $k$ verifies $S\cap K'=\emptyset$ or $|K'\setminus S|\leq k$. 
Let $M'$ be the set of edges of $M$ having an endpoint in $K'$.
We denote by $L'$ the set of vertices $M'\setminus K'$, \ie{}the endpoints
of edges in $M'$ which do not belong to $C\setminus x$. Now let us consider the 
graph $G':=G\setminus (C\setminus x)$ with root $x$.
The set $L'$ has size at least $h(2k+1)$, thus by Theorem~\ref{extractright} 
there is a vertex $y$ in $L'$ such that whenever we delete $k$ edges 
in $G$ such that at most $k$ vertices of $L'$ belong to the component
of $x$, then $y$ does not belong to the component of $x$.
The vertex $y$ being an element of $L'$, we consider the request
$zy\in M'$, where $z$ belongs to $V(C\setminus x)$. 

We claim that the request $zy$ is irrelevant.
Indeed, let $F$ be a multicut of $R\setminus zy$ with at most $k$ edges 
such that $F_C=F\cap C$ is the border of a left cut. Let $S$ be the 
component of $x$ in $C\setminus F_C$. The set $S$ is a left cut and has border of size
at most $k$, hence either $S$ completely isolates $x$ from $K'$ or $S$
isolates at most $k$ vertices of $K'$ from $x$. If $K'$ is isolated from
$x$, we have in particular that $x$ is disconnected from $y$, hence
the request $zy$ is cut by $F$. So we assume that a subset $K''$ containing all but at most $k$ vertices 
of $K'$ is included in $S$. Hence, denoting by $L''$ the 
other endpoints of the edges of $M'$ intersecting $K''$, this means
that $F$ must disconnect $x$ from $L''$. Therefore the set $F$ of at most $k$ edges
disconnects $x$ from at most $k+1$ elements of $L'$ (the $k$ elements of $L''$
and possibly $y$), so by definition of $y$, the set $F$ disconnects 
$x$ from $y$. In particular $zy$ is cut by $F$. Thus the request $zy$ is indeed irrelevant. All the computations so far are FPT.

We repeat this process, removing irrelevant requests until the size of $L$ does not exceed 
$b(k)$. We then set $K(C):=L$, and the conclusion of Theorem~\ref{cherryrequestreduction} holds.
\end{proof}

Let $C$ be a cherry of a graph $G$ with deletion allowance $k$. A subset 
${\mathcal L}$ of the edges of $C$ is \emph{active} when, if a multicut $F$ 
of size at most $k$ exists, then there exists a multicut $F'$ of size at most $|F|$ such that 
$F'\setminus C = F\setminus C$ and $F'\cap C \subseteq {\mathcal L}$. When 
the set ${\mathcal L}$ is clear from the context, we say by extension that 
edges of ${\mathcal L}$ are \emph{active}.

\begin{lemma} \label{finitelycut}
Let $C$ be an $x$-cherry of a graph $G$ with deletion allowance $k$, and 
let $K$ be the set of all terminals of $C\setminus x$. Let ${\mathcal L}(C)$ be the 
union of all borders of cuts of $C^y_{\leq k}$, where $y\in K$. Then 
${\mathcal L}(C)$ is active, and has size at most $k|K|k!$.
\end{lemma}

\begin{proof}
Assume that $F$ is a multicut with size at most $k$. Let $S$ be the component 
of $x$ in $C\setminus F$. Let $T$ be a left cut with $T\subseteq S$ and $\delta(T)\le \delta(S)$. 
If a component $U$ of 
$\overline T$  does not intersect $K$, the set $\Delta(T)\setminus \Delta (U)$ is still a multicut.
Finally, we can assume that all components $U$ of $\overline T$ intersect $K$, in which
case $\Delta (U)\in C^y_{\leq k}$ for some $y$ in $K$, hence $\Delta (T)$ 
is included in ${\mathcal L}$. The set $F'=F\setminus C \cup \Delta(T)$ is a 
multicut, and the size bound for ${\mathcal L}$ follows from Theorem~\ref{boundsize}.
\end{proof}

\begin{theorem}\label{boundedcut}
Let $H_1,H_2,\dots ,H_p$ be $x$-cherries of a graph $G$ with deletion allowance $k$ such that $H_1\setminus x,H_2\setminus x,\dots ,H_p\setminus x$ 
are pairwise disjoint. Assume that for every $i$, $U_i:=H_1\cup \dots \cup H_i$ is a cherry.
Then every $U_i$ has a bounded active set ${\mathcal L}_i$ such that  $\mathcal{L}_j\cap U_i\subseteq \mathcal{L}_i$
whenever $i\leq j$.
\end{theorem}

\begin{proof}
By Theorem~\ref{cherryrequestreduction}, we can reduce the set of terminals in $U_1$ to a 
bounded set $K_1$. The set ${\mathcal L}_1={\mathcal L}(U_1)$ is bounded and active by Lemma~\ref{finitelycut}.
The requests of $C_1\setminus K_1$ are irrelevant in $U_2$ since they are irrelevant in $U_1$, hence 
we can assume that Theorem~\ref{cherryrequestreduction} applied to $U_2$ yields a set of terminals $K_2\subseteq K_1\cup C_2$.
Let ${\mathcal L}_2$ be the active edges associated to $K_2$. Note that if an edge
$e\in {\mathcal L}_2$ is in $U_1$, it must belong to some $C^y_{\leq k}$ for some 
$y\in K_2\cap U_1$. Since $K_2\subseteq K_1\cup C_2$, we have $y\in K_1$, hence $e\in {\mathcal L}_1$, 
which is the property we are looking for. We extract $K_3$ from $K_2\cup C_3$, and iterate 
our process to form the sequence $\mathcal{L}_i$.
\end{proof}

%%%%%%%%%%%%%%%%%%%%%%%%%%%%%%%%%%%%%%%%%%%%%%%%%%%%%%%%%%%%%%%%%%%%%%%%%%%%%%%%%%%%%%%%%%%%%%%%%%%%%%%%%%%%%%%%%%%%%%%%%
%%%%%%%%%%%%%%%%%%%%%%%%%%%%%%%%%%%%%%%%%%%%%%%%%%%%%%%%%%%%%%%%%%%%%%%%%%%%%%%%%%%%%%%%%%%%%%%%%%%%%%%%%%%%%%%%%%%%%%%%%
%%%%%%%%%%%%%%%%%%%%%%%%%%%%%%%%%%%%%%%%%%%%%%%%%%%%%%%%%%%%%%%%%%%%%%%%%%%%%%%%%%%%%%%%%%%%%%%%%%%%%%%%%%%%%%%%%%%%%%%%%
%%%%%%%%%%%%%%%%%%%%%%%%%%%%%%%%%%%%%%%%%%%%%%%%%%%%%%%%%%%%%%%%%%%%%%%%%%%%%%%%%%%%%%%%%%%%%%%%%%%%%%%%%%%%%%%%%%%%%%%%%
%%%%%%%%%%%%%%%%%%%%%%%%%%%%%%%%%%%%%%%%%%%%%%%%%%%%%%%%%%%%%%
%%%%%%%%%%%%%%%%%%%%%%%%%%%%%%%%%%%%%%%%%%%%%%%%%%%%%%%%%%%%%%
%%%%%%%%%%%%%%%%%%%%%%%%%%%%%%%%%%%%%%%%%%%%%%%%%%%%%%%%%%%%%%
%%%%%%%%%%%%%%%%%%%%%%%%%%%%%%%%%%%%%%%%%%%%%%%%%%%%%%%%%%%%%%%%%%%%%%%%%%%%%%%%%%%%%%%%%%%%%%%%%%%%%%%%%%%%%%%%%%%%%%%%%%%%
%%%%%%%%%%%%%%%%%%%%%%%%%%%%%%%%%%%%%%%%%%%%%%%%%%%%%%%%%%%%%%
%%%%%%%%%%%%%%%%%%%%%%%%%%%%%%%%%%%%%%%%%%%%%%%%%%%%%%%%%%%%%%
%%%%%%%%%%%%%%%%%%%%%%%%%%%%%%%%%%%%%%%%%%%%%%%%%%%%%%%%%%%%%%%%%%%%%%%%%%%%%%%%%%%%%%%%%%%%%%%%%%%%%%%%%%%%%%%%%%%%%%%%%%%%
%%%%%%%%%%%%%%%%%%%%%%%%%%%%%%%%%%%%%%%%%%%%%%%%%%%%%%%%%%%%%%
%%%%%%%%%%%%%%%%%%%%%%%%%%%%%%%%%%%%%%%%%%%%%%%%%%%%%%%%%%%%%%
%%%%%%%%%%%%%%%%%%%%%%%%%%%%%%%%%%%%%%%%%%%%%%%%%%%%%%%%%%%%%%

\section{Reducing \M{} to \CM{}}\label{CM}

Let $G=(V,E)$ be a connected graph, and $R$ be a set of requests. A \emph{vertex-multicut} $Y$
is a subset of $V$ such that every $xy$-path of $G$ where $xy\in R$
contains an element of $Y$. Let $A$ be a connected component of 
$G\setminus Y$.  We call \emph{$Y$-component}, or \emph{component}, the union 
of $A$ and its set of neighbors in $Y$. Let $C$ be a component, the vertices
of $C\cap Y$ are the \emph{attachment vertices} of $C$.

\subsection{Component Multicut}
Here is our first intermediate problem. 

\begin{quote}
\CM{}:\\
\textbf{Input}: A connected graph $G=(V,E)$, a vertex-multicut $Y$, a set of 
requests $R$, a set $I$ of $q$ integers such that:
\begin{enumerate}

\item There are $q$
$Y$-components $G_1,\dots ,G_q$ with two attachment vertices $x_i,y_i$. The other components have only one attachment vertex.

\item Every $G_i$ has a $x_iy_i$-path $P_i$ called \emph{backbone} of $G_i$. 
Moreover, the deletion of every edge of 
$P_i$ decreases the edge connectivity in $G_i$ between $x_i$ and $y_i$.

\item The set $I$ consists of $q$ non negative integers $f_1,\dots ,f_q$ such that 
$f_1+\dots +f_q\leq k-q$.

\end{enumerate}
\textbf{Parameter}: $k$.\\
\textbf{Output}: TRUE if there exists a multicut $F$ such that:
\begin{enumerate}
\item every path $P_i$ contains exactly one edge of $F$,

\item every $G_i$ contains exactly $1+f_i$ edges of $F$,

\item the solution $F$ \emph{splits} $Y$, \ie{}the connected components of $G\setminus F$ contain at most 
one vertex of $Y$.

\end{enumerate}
Otherwise, the output is FALSE. \newline
\end{quote}

We call $Y$ a \emph{cutset}, \ie{}a vertex-multicut 
which must be split by the solution $F$.
The edges of $G$ which do not belong to the backbones are called
\emph{free edges}. The backbone $P_i$, in which only one edge
is deleted, is the crucial structure of $G_i$. Indeed, the whole proof
consists of modifying $G_i$ step by step to finally completely 
reduce it to the backbone $P_i$.
Here $f_i$ is the number of free edges one can delete
in $G_i$. Observe that $k-q-f_1-\dots -f_q$ free edges can be deleted in 
components with one attachment vertex. Our first reduction 
is the following:

\begin{theorem}\label{multitomultiwithback}
\M{} can be reduced to \CM{} in FPT time.
\end{theorem}

The remaining of Section~\ref{CM} is devoted to the proof of Theorem~\ref{multitomultiwithback}. 
We first construct a cutset $Y$.
Then we prove that we can reduce to $Y$-components with
one or two attachment vertices. Finally, we show that we can assume that every component with two 
attachment vertices has a path in which exactly one edge is chosen in the solution.
This is our backbone.

\subsection{The cutset $Y$.}

We use iterative compression to prove the existence of a 
vertex-multicut of size $k+1$, as was done in~\cite{MulticutFPTapprox}.

\begin{lemma}\label{coco}
\M{} can be solved in time $O(f(k)n^c))$ if \M{} given a vertex 
multicut of size at most $k+1$ can be solved in time $O(f(k)n^{c-1})$.
\end{lemma}

\begin{proof} 
By induction on $n$, we solve \M{} in 
time $f(k)(n-1)^c$ on $G-v$ where $v\in V(G)$.
If the output is FALSE, we return it, otherwise the output is
a multicut $F$ of size at most $k$. Let $X$ be a vertex cover
of $F$ of size at most $k$. Thus $X\cup \{v\}$ is a vertex-multicut
of our instance, so we solve \M{} in time $f(k)n^{c-1}+f(k)(n-1)^c$
which is at most $f(k)n^{c}$.
\end{proof}

So we can assume that we have a vertex-multicut $Y$ of size at most $k+1$.

\begin{lemma}\label{caca}
We can assume that the solution $F$ splits $Y$.
\end{lemma}

\begin{proof} 
To any solution $F$ is associated the partition of $G\setminus F$ into
connected components. In particular this induces a partition of $Y$.
We branch over all possible partitions of $Y$. In a given
branch, we simply contract the elements of $Y$ belonging to the same part
of the partition corresponding to the branch.
\end{proof}

During the following reduction proof, the set $Y$ will never decrease. Since
one needs $k+1$ edges to separate $k+2$ vertices, the size of $Y$
cannot exceed $k+1$, otherwise we return FALSE. Hence our
primary invariant is the size of $Y$, and we immediately conclude if
we can make it increase. 

%%%%%%%%%%%%%%%%%%%%%%%%%%%%%%%%%%%%%%%%%%%%%%%%%%%%%%%%%%%%%%%%%%%%%%%%
%%%%%%%%%%%%%%%%%%%%%%%%%%%%%%%%%%%%%%%%%%%%%%%%%%%%%%%%%%%%%%%%%%%%%
%%%%%%%%%%%%%%%%%%%%%%%%%

\subsection{Reducing attachment vertices.}
Our second invariant, which we intend to maximize, is the 
number of $Y$-components with at least two
attachment vertices. This number cannot exceed $k$, since any solution splits $Y$. 
Our third invariant is the sum of the edge connectivity between
all pairs of vertices of $Y$, which we want to increase. This invariant is 
bounded by $k{|Y|\choose 2}$ since the connectivity between two elements 
of $Y$ is at most $k$. Note that this invariant never decreases when
we contract vertices. 

\begin{lemma}\label{reduceattach}
If $C$ is a $Y$-component with at least three attachment
vertices, we improve our invariant.
\end{lemma}

\begin{proof}
Let $x,y,z$ be attachment vertices of $C$.
Let $\lambda $ be the edge-connectivity between $x$ and $y$ in $C$.
Let $P_1,\dots,P_{\lambda}$ be a set of edge-disjoint $xy$-paths. 
A \emph{critical edge} is an edge which belongs to some $xy$-edge cut of
size $\lambda$. Note that every critical edge belongs to some path $P_i$.
A \emph{slice} of $C$ is a connected 
component of $C$ minus the critical edges. Given a vertex $v$ of 
$C$, the \emph{slice of $v$}, denoted by $SL(v)$, is the slice of $C$ containing $v$.
Let $B(z)$ be
the \emph{border} of $SL(z)$, \ie{}the set of vertices of $SL(z)$ which 
are incident to a critical edge. Note that $B(z)$ intersects every path 
$P_i$ on at most two vertices, namely the leftmost vertex of $P_i$
belonging to $SL(z)$ and the rightmost vertex of $P_i$
belonging to $SL(z)$. In particular, $B(z)$ has $b$ vertices, where $b\leq 2\lambda$.

We branch over $b+1$ choices to decide whether one of the $b$ 
vertices of $B(z)$ belongs to a component of $G\setminus F$ (where $F$
is our solution) which does not contain a vertex of $Y$. 
When this is the case,
the vertex is added to $Y$, hence we increase our primary
invariant. In the last branch, all the vertices of $B(z)$
are connected to a vertex of $Y$ in $G\setminus F$.
We branch again over all mappings $f$ from $B(z)$
into $Y$. In each branch, the vertex $v\in B(z)$ is connected
to $f(v)\in Y$ in $G\setminus F$. Hence we can contract 
every vertex $v\in B(z)$ to the vertex $f(v)$ of $Y$. 
This gives a new graph $G'$.
We denote by $S'$ the subgraph $SL(z)$ after contraction 
of the vertices of $B(z)$. Observe that $S'$ is a $Y$-component
of $G'$.

If $x$ and $y$ belong to $S'$, the edge connectivity  between $x$ and $y$ has increased. Indeed, there is now a path $P$  joining $x$ to $y$ inside $S'$, in particular $P$ has no critical edge.  Thus the connectivity between $x$ and $y$ has increased, so our invariant has improved. We assume without loss of generality that $x$ does not belong to $S'$.

If $S'$ contains an element of $Y$ distinct from $z$,
then $S'$ is a $Y$-component with at least two attachment vertices.
Moreover, there exists a path $P$ in $C\setminus S'$ from $x$ to $B(z)$. Hence we have created
an extra $Y$-component with at least two attachment vertices in $G'$, which improves our second invariant.

In our last case, $S'$ only intersects $Y$ on $z$. Therefore
$B(z)$ is entirely contracted to $z$. In particular 
$z$ is now incident to a critical edge $e$. So there exists an 
$xy$-cut $A$ with $\delta(A)=\lambda$ and $e\in \Delta(A)$. Without
loss of generality, we assume that $z\notin A$ (otherwise we consider 
the $yx$-cut $\overline A$). We denote by $B$ the vertices of $\overline A$
with a neighbor in $A$. In particular, $B$ contains $z$, has size 
at most $\lambda$, and every $xy$-path in $C$ contains a vertex
of $B$. Let us denote by $L$ the set $A\cup B$ and by $R$ the set $\overline A$.
Note that $L\cap R=B$. We now branch to decide 
in which components of $G\setminus F$ the elements of $B$ are 
partitioned. If an element of $B$ is not connected to $Y$ in $G\setminus F$, 
we improve our invariant. If each element of $B$ is contracted 
to a vertex of $Y$, both $L$ and $R$ in the contracted graph are
$Y$-components with at least two attachment vertices (respectively $\{x,z\}$ and $\{y,z\}$). We again
improve our invariant.
\end{proof}

%%%%%%%%%%%%%%%%%%%%%%%%%%%%%%%%%%%%%%%%%%%
%%%%%%%%%%%%%%%%%%%%%%%%%%%%%%%%%%%%%%%%%%%
%%%%%%%%%%%%%%%%%%%%%%%%%%%%%%%%%%%%%%%%%%%
%%%%%%%%%%%%%%%%%%%%%%%%%%%%%%%%%%%%%

\subsection{Backbones.} We now assume that every
component has at most two attachment vertices. 
Let $G_1,\dots ,G_q$ be the components of $G$ with two attachment 
vertices. We denote by $\lambda _i$ the edge connectivity 
of $G_i$ between its two attachment vertices $x_i,y_i$.
Recall that our third invariant is just the sum of the $\lambda _i$.

\begin{lemma}\label{degreexy}
We can assume that $x_i$ and $y_i$ have degree $\lambda _i$ in $G_i$.
\end{lemma}

\begin{proof}
Let $A$ be the unique left $x_iy_i$-cut with $\delta (A)=\lambda _i$ in the graph $G_i$ rooted in $x_i$. 
Let $B$ be the set of vertices of $A$
with a neighbor in $\overline A$. We now branch to decide 
how the components of $G\setminus F$ partition $B$. If an element of $B$ is not connected to $Y$ in $G\setminus F$, 
we improve our invariant. 
If an element of $B$ is contracted to $y_i$, we increase
$\lambda _i$. Hence all elements of $B$ are contracted to $x_i$.
Therefore $A$ becomes an $x_i$-cherry, hence $A\setminus x_i$
is removed from $G_i$. The degree of $x_i$ inside $G_i$ is now exactly
$\lambda _i$. We apply the same argument to reduce the degree of $y_i$ to $\lambda _i$.
\end{proof}

We now branch over all partitions of $k$ into $k_0+k_1+\dots +k_q=k$,
where $k_i$ is the number of edges of our solution chosen in
$G_i$ when $i>0$, and $k_0$ is the number of edges chosen in the 
$y$-components for $y\in Y$.

\begin{lemma}\label{backbone}
Every component $G_i$ can be deleted or has a backbone.
\end{lemma}

\begin{proof} 
If $k_i\geq 2\lambda _i$, we simply delete in $G_i$ the edges incident 
to $x_i$ and $y_i$, reduce our parameter, and improve our invariant. 
So we can assume that $k_i\leq 2\lambda_i -1$. Let $P_1,P_2,\dots ,P_{\lambda _i}$ 
be some edge-disjoint $x_iy_i$-paths.
Our algorithm now branches $2\lambda _i$ times, where the branches are called $B_j$ and $B'_j$ 
for $j=1,\dots ,\lambda _i$. In the branch $B_j$, we assume that there is only one
edge of our solution selected in $P_j$, and this edge belongs to an $x_iy_i$-cut of size $\lambda_i$. In the branch $B'_j$,
we assume that all the edges of our solution selected in $P_j$ are not critical.
Let us show that any solution $F$ belongs to one of these branches. If $F$ does not
belong to any branch $B'_j$, this means that $F$ uses at least one critical edge in each 
$P_j$. But since $k_i\leq 2\lambda_i -1$ some $P_j$ only intersects $F$ on one 
edge, which is therefore critical. Hence $F$ is a solution in the branch $B_j$. Thus this 
branching process is valid. In the branch $B_j$, we contract all non critical edges
of $P_j$, therefore $P_j$ is the backbone we are looking for.  In the branch $B'_j$, we contract all critical edges
of $P_j$, hence the connectivity $\lambda _i$ increases. We thus improve our invariant.
\end{proof}

This concludes the proof of Theorem \ref{multitomultiwithback}.

%%%%%%%%%%%%%%%%%%%%%%%%%%%%%%%%%%%%%%%%%%%%%%%%%%%%%%%%%%%%%%%%%%%%%%%%%%%%%%%%%%%%%%%%%%%%%%%%%%%%%%%%%%%%%%%%%%%%%%%%%
%%%%%%%%%%%%%%%%%%%%%%%%%%%%%%%%%%%%%%%%%%%%%%%%%%%%%%%%%%%%%%%%%%%%%%%%%%%%%%%%%%%%%%%%%%%%%%%%%%%%%%%%%%%%%%%%%%%%%%%%%
%%%%%%%%%%%%%%%%%%%%%%%%%%%%%%%%%%%%%%%%%%%%%%%%%%%%%%%%%%%%%%%%%%%%%%%%%%%%%%%%%%%%%%%%%%%%%%%%%%%%%%%%%%%%%%%%%%%%%%%%%
%%%%%%%%%%%%%%%%%%%%%%%%%%%%%%%%%%%%%%%%%%%%%%%%%%%%%%%%%%%%%%%%%%%%%%%%%%%%%%%%%%%%%%%%%%%%%%%%%%%%%%%%%%%%%%%%%%%%%%%%%
%%%%%%%%%%%%%%%%%%%%%%%%%%%%%%%%%%%%%%%%%%%%%%%%%%%%%%%%%%%%%%
%%%%%%%%%%%%%%%%%%%%%%%%%%%%%%%%%%%%%%%%%%%%%%%%%%%%%%%%%%%%%%
%%%%%%%%%%%%%%%%%%%%%%%%%%%%%%%%%%%%%%%%%%%%%%%%%%%%%%%%%%%%%%
%%%%%%%%%%%%%%%%%%%%%%%%%%%%%%%%%%%%%%%%%%%%%%%%%%%%%%%%%%%%%%%%%%%%%%%%%%%%%%%%%%%%%%%%%%%%%%%%%%%%%%%%%%%%%%%%%%%%%%%%%%%%
%%%%%%%%%%%%%%%%%%%%%%%%%%%%%%%%%%%%%%%%%%%%%%%%%%%%%%%%%%%%%%
%%%%%%%%%%%%%%%%%%%%%%%%%%%%%%%%%%%%%%%%%%%%%%%%%%%%%%%%%%%%%%
%%%%%%%%%%%%%%%%%%%%%%%%%%%%%%%%%%%%%%%%%%%%%%%%%%%%%%%%%%%%%%
%%%%%%%%%%%%%%%%%%%%%%%%%%%%%%%%%%%%%%%%%%%%%%%%%%%%%%%%%%%%%%
%%%%%%%%%%%%%%%%%%%%%%%%%%%%%%%%%%%%%%%%%%%%%%%%%%%%%%%%%%%%%%
%%%%%%%%%%%%%%%%%%%%%%%%%%%%%%%%%%%%%%%%%%%%%%%%%%%%%%%%%%%%%%
%%%%%%%%%%%%%%%%%%%%%%%%%%%%%%%%%%%%%%%%%%%%%%%%%%%%%%%%%%%%%%
\section{\BM{} is FPT}\label{trlemon}
\subsection{Backbone Multicut}
We introduce here the problem \BM{}, which is a generalization of \CM{}. Our 
goal is to show that \BM{} is solvable in FPT time, which implies that \CM{} is FPT, which in turns implies that \M{} is FPT thanks to Theorem~\ref{multitomultiwithback}.

\begin{quote}
\BM{}:\\
\textbf{Input}: A connected graph $G=(V,E)$, a set of half-requests $R$, a set $Y$ of vertices, 
a set $B$ of $q$ variables, a set $\mathcal C$ of clauses, a set $I$ of $q$ integers such that:
\begin{enumerate}
\item $G$ has $q$ $Y$-components called $G_i$ with two attachment vertices $x_i,y_i\in Y$, with $i=1,\dots ,q$. Moreover,$G_i$ has a backbone $P_i$ and the $x_i,y_i$-connectivity in $G_i$ is $\lambda_i$.
Recall that the edges of $G$ which are not in backbones are called \emph{free edges}.

\item The set $R$ contains \emph{half-requests}, \ie{}sets of triples $(u,y,v)$, informally meaning
that vertex $u$ sends a request to vertex $v$ via $y$, where $y\in Y$. Also, $Y$ is a $u,v$-cut for every half-request $(u,y,v)\in R$.

\item The set $B$ contains $q$ integer-valued variables $c_1,\dots ,c_q$. Each variable 
$c_i$ corresponds to the deletion of one edge in the backbone $P_i$.
Formally, if the edges of $P_i$ are $e_1,\dots ,e_{\ell _i}$, ordered from
$x_i$ to $y_i$, the variable
$c_i$ can take all possible values from $1$ to $\ell _i$, and 
$c_i=r$ means that we delete the edge $e_r$ in $P_i$. 

\item The clauses in $\mathcal C$ have four possible types: $(c_i\leq a \Rightarrow c_j\leq b)$, or
$(c_i\leq a \Rightarrow c_j\geq b)$, or $(c_i\geq a \Rightarrow c_j\geq b)$, or $(c_i\geq a \Rightarrow c_j\leq b)$.

\item The set $I$ consists of $q$ non negative integers $f_1,\dots ,f_q$ summing to a value
at most $k$. Each integer $f_i$ corresponds to the number of free edges of the solution 
which are chosen in $G_i$.

\end{enumerate}
\textbf{Parameter}: $k$.\\
\textbf{Output}: TRUE if:
\begin{enumerate}
\item there exists an assignment of the variables of $B$ which satisfies 
${\mathcal C}$,

\item there exists a subset $F$ of at most $k$ free edges of $G$,

\item for each $i=1,\dots ,q$, the set $F$ contains $f_i$ free edges in $G_i$,

\item the union $F'$ of $F$ and the backbone edges corresponding to the variables 
of $B$ splits $Y$ and intersects every half-request of $R$, \ie{}for every half-request $(u,y,v)\in R$ every path between $u$ and $v$ containing $y$ intersects $F'$.
\end{enumerate}
Otherwise, the output is FALSE. \newline
\end{quote}

Note that the deletion allowance of 
\BM{} is $k+q$. \CM{} directly translates into \BM{} with an empty set 
of clauses, and where each request is simulated by one or two half-requests. This section is devoted to the proof of the following result.

\begin{theorem}\label{redmultpaths}
\BM{} can be solved in FPT time.
\end{theorem}

%%%%%%%%%%%%%%%%%%%%%%%%%%%%%%%%%%%%%%%%%%%%%%%%%%%%%%%%%%%%%%
%%%%%%%%%%%%%%%%%%%%%%%%%%%%%%%%%%%%%%%%%%%%%%%%%%%%%%%%%%%%%%
%%%%%%%%%%%%%%%%%%%%%%%%%%%%%%%%%%%%%%%%%%%%%%%%%%%%%%%%%%%%%%
%%%%%%%%%%%%%%%%%%%%%%%%%%%%%%%%%%%%%%%%%%%%%%%%%%%%%%%%%%%%%%

\subsection{Invariants} 
Our primary invariant is the sum of the $f_i$, which starts with value at most $k$
and is nonnegative. Any branch in which we can decrease 
it will be considered solved. Our secondary invariant is the sum of the $\lambda _i-1$,
called the \emph{free connectivity}, which we try to increase. Observe that this invariant is bounded above by $k$.
For our last invariant, recall that the 
\emph{slice} $SL(v)$ of some vertex $v$ in a component $G_i$ is the connected 
component containing $v$ of $G_i$ minus its \emph{critical} edges, \ie{}edges of $\lambda_i$-cuts. Observe that since 
the edges of $P_i$ are critical, the slices of distinct vertices in $P_i$ do not intersect.
The \emph{slice connectivity} of a vertex $v$ in $P_i$ is the 
$x_iy_i$-edge-connectivity of $G_i\setminus SL(v)$.
We denote it by $sc(v)$. For instance, if the set of neighbors of $v$ intersect 
every $x_iy_i$-path in $G_i\setminus P_i$ then we have $sc(v)=0$.  Conversely,
if $v\in P_i$ has only neighbors in $P_i$, we have $sc(v)=\lambda_i -1$.
The \emph{slice connectivity} $sc_i$ of $G_i$ is the maximum of $sc(v)$, where 
$v\in P_i$. Our third invariant is the sum $sc$ of the $sc_i$, for $i=1,\dots ,q$,
and we try to minimize this invariant. Observe that at any step,
$sc$ is at most $k$.

\begin{figure}
\center\includegraphics[scale=.75]{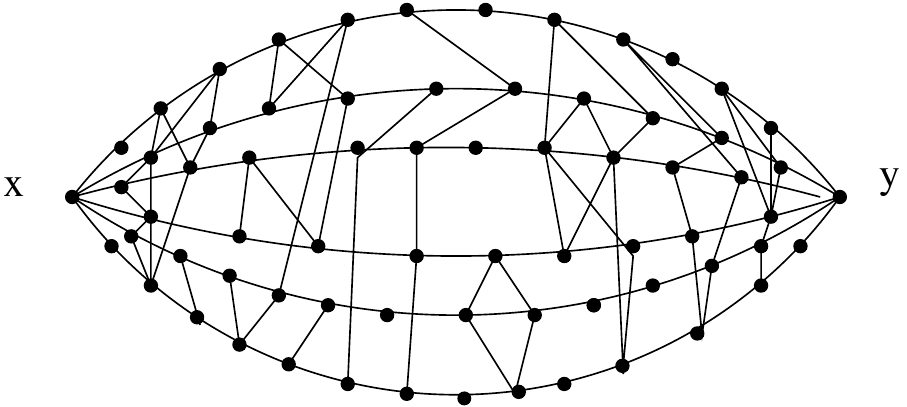}
\caption{A component in a \BM{} instance. The path at the bottom is the backbone. Note that each edge of the backbone does indeed belong to a $\lambda$-cut.}
\label{lemon}
\end{figure}

Our goal is to show that we can always improve our invariant, or conclude that 
$\lambda _i=1$ for all $i$. We consider a component $G_i$ with $\lambda _i>1$, say $G_1$.
To avoid cumbersome indices, we assume that the attachment vertices of $G_1$ 
are $x$ and $y$, and that their edge-connectivity is denoted by $\lambda $ instead of $\lambda _1$.
Moreover, we still denote by $P_1$ the backbone of $G_1$, but we assume (against our previous notations)
that  $P_1,P_2,\dots ,P_{\lambda}$ is a set of edge-disjoint $xy$-paths in $G_1$. 
We visualize $x$ to the left and $y$ to the right (see Figure~\ref{lemon}). Hence when we say that a vertex $u\in P_i$
is to the left of some $v\in P_i$, we mean that $u$ is between $x$ and $v$ on $P_i$.

%%%%%%%%%%%%%%%%%%%%%%%%%%%%%%%%%%%%%%%%%%%%%%%%%%%%%%%%%%%%%%
%%%%%%%%%%%%%%%%%%%%%%%%%%%%%%%%%%%%%%%%%%%%%%%%%%%%%%%%%%%%%%
%%%%%%%%%%%%%%%%%%%%%%%%%%%%%%%%%%%%%%%%%%%%%%%%%%%%%%%%%%%%%%
%%%%%%%%%%%%%%%%%%%%%%%%%%%%%%%%%%%%%%%%%%%%%%%%%%%%%%%%%%%%%%
\subsection{Contracting edges} 
In our proof, we contract edges of the backbone and free edges 
which are not critical. At any step, we always preserve the fact 
that the edges of the backbone are critical.

When contracting an edge of the backbone $P_1$, we need to 
modify several parameters.
Assume that the edges of $P_1$ are $e_1,\dots ,e_{\ell}$. The variable $c_1$
represents the edge of $P_1$ which is deleted
in our multicut. Now assume that the edge $e_i=v_iv_{i+1}$ is contracted.
All the indices of the edges which are at least $i+1$ are decreased by one.
All the constraints associated to the other backbones are not 
affected by the transformation. 
However, each time a clause contains a litteral $c_1 \geq j$, where $j >i$, this litteral must 
be replaced by $c_1 \geq j-1$. Similarly each occurence of $c_1 \leq j'$ for $j' \ge i$ 
must be replaced by $c_1 \leq j'-1$. 
If a set of edges is contracted, we perform the contractions one by one.

The collection of paths $P_2,\dots ,P_{\lambda}$ can be affected 
during our contractions since it can happen that a path 
$P_i$ with $i\geq 2$ contains both endpoints of a contracted 
edge $uv$. In such a case, we remove from $P_i$
the loop formed by the contraction, \ie{}the subpath of 
$P_i$ between $u$ and $v$. We thus preserve our path collection.

%%%%%%%%%%%%%%%%%%%%%%%%%%%%%%%%%%%%%%%%%%%%%%%%%%%%%%%%%%%%%%
%%%%%%%%%%%%%%%%%%%%%%%%%%%%%%%%%%%%%%%%%%%%%%%%%%%%%%%%%%%%%%
%%%%%%%%%%%%%%%%%%%%%%%%%%%%%%%%%%%%%%%%%%%%%%%%%%%%%%%%%%%%%%
%%%%%%%%%%%%%%%%%%%%%%%%%%%%%%%%%%%%%%%%%%%%%%%%%%%%%%%%%%%%%%
\subsection{Choosing a stable edge} 

Let $v$ be a vertex of $P_1$. The \emph{tag} of $v$ is the subset 
$t(v):=\{i~:~P_i\cap SL(v)\neq \emptyset\}$, \ie{}the set of indices of the paths 
intersecting the slice of $v$. Note that
$t(v)$ contains 1. Observe also that the slice connectivity of 
$G_1$ is the maximum  of $\lambda -|t(v)|$, where $v$ belongs to $P_1$.
By extension, the \emph{tag} of an edge $v_iv_{i+1}$ of the backbone $P_1$ 
is the ordered pair $(t(v_i),t(v_{i+1}))$. When speaking of an \emph{$XY$-edge}, 
we implicitely mean that its tag is $(X,Y)$. In particular, the edge 
of $P_1$ which is selected in our solution has a given tag.
We branch over the possible choices for the tag $XY$ of the deleted edge
of $P_1$. Let us assume that the chosen edge has tag $XY$.

\begin{lemma}\label{assosetdisj}
If $X\neq Y$, we improve our invariant.
\end{lemma}

\begin{proof}
Since only one edge is cut in the backbone, we can contract all the edges of $P_1$ with tags different from $XY$. 
Observe that when contracting some $UV$-edge of $P_1$, the tag of
the resulting vertex contains $U \cup V$ since the slice of the resulting 
vertex contains the union of both slices (it can actually be larger).
After our contraction, all the edges of $P_1$ between two consecutive occurences 
of $XY$-edges are contracted, hence the tag of every vertex of $P_1$ 
now contains $X\cup Y$. In particular, the slice connectivity of $G_1$ decreases while the free
connectivity is unchanged. Thus our invariant improved.
\end{proof}

Therefore we may assume that we choose an $XX$-edge in our solution. Let us 
contract all the edges of $P_1$ which are not $XX$-edges.
By doing so, we now have that the tag of every vertex of $P_1$ contains $X$.
After this contraction, our instance is modified, hence we have
to branch again over the choice of the tag of the edge
chosen in our solution. Any choice different from $XX$ 
increases the slice connectivity. Hence we can still assume that the 
tag of the chosen edge is $XX$. 

The slice connectivity of $G_1$ is $\lambda -|X|$. An $XX$-edge $uv$ of the backbone 
is \emph{unstable} if, when contracting $uv$, the tag of the vertex $u=v$ 
increases (\ie{}strictly contains $X$). Otherwise $uv$ is \emph{stable}.
We branch on the fact that the chosen $XX$-edge is stable or unstable. 

\begin{lemma}\label{cici}
If the chosen $XX$-edge is unstable, we improve the invariant.
\end{lemma}

\begin{proof} 
We enumerate the set of all unstable edges from left to right along $P_1$, and partition them according 
to their index into the odd indices and the even indices. We branch according to
the index of the chosen unstable edge. Assume for instance that the chosen unstable edge has
odd index. We contract all the edges of $P_1$ save the odd unstable edges. 
We claim that every vertex tag of backbone vertices now strictly contains 
$X$. Indeed, all edges of $P_1$ between two consecutive odd unstable edges are contracted,
in particular some even unstable edge. Thus, since this even edge
is unstable, the tag now strictly contains $X$. Hence the slice connectivity decreases.
\end{proof}

%%%%%%%%%%%%%%%%%%%%%%%%%%%%%%%%%%%%%%%%%%%%%%%%%%%%%%%%%%%%%%
%%%%%%%%%%%%%%%%%%%%%%%%%%%%%%%%%%%%%%%%%%%%%%%%%%%%%%%%%%%%%%
%%%%%%%%%%%%%%%%%%%%%%%%%%%%%%%%%%%%%%%%%%%%%%%%%%%%%%%%%%%%%%
%%%%%%%%%%%%%%%%%%%%%%%%%%%%%%%%%%%%%%%%%%%%%%%%%%%%%%%%%%%%%%

\subsection{Contracting slices}
In this part, we assume that the chosen edge of $P_1$ is a stable $XX$-edge. 
A vertex $v$ of $P_1$ is \emph{full} if $v$ belongs to every $P_i$, where $i\in X$ (see Figure~\ref{full}).
Our goal in this subsection is to show that we can reduce to the case 
where $X=\{1,\dots ,\lambda\}$.
By the previous section, any branching 
increasing the tag of the chosen edge would improve the 
invariant. So we assume that in all our branchings, the
chosen edge is still a stable $XX$-edge.

%\begin{comment}
%%%% FIGURES !! %%%%%%%

\begin{figure}\center\includegraphics[scale=.75]{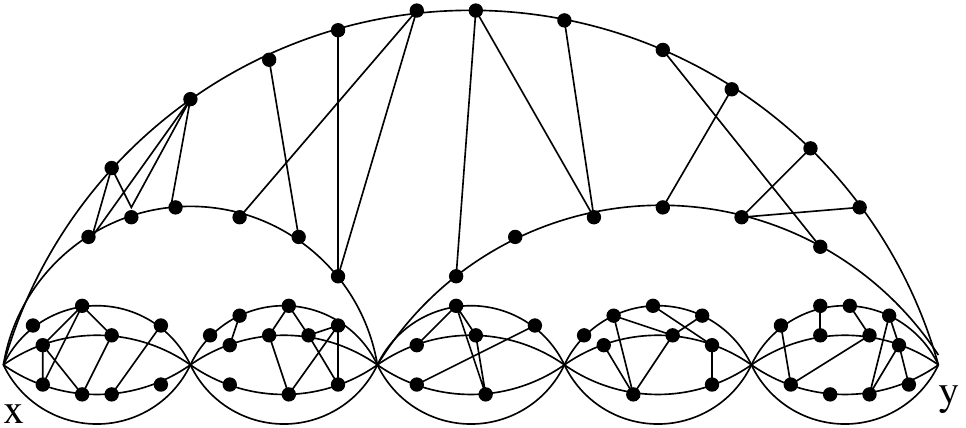}\caption{All backbone vertices of this component are full, where $X=\{1,2,3,4\}$, with the backbone $P_1$ at the bottom.}\label{full}\end{figure}

%\end{comment}

\begin{lemma}\label{seqlemons}
We can assume that all backbone vertices are full.
\end{lemma}

\begin{proof} 
We can first assume that there are at most $k$ vertices with tag $X$
between two full vertices.
Indeed, let us enumerate $w_1,w_2,\dots $ the vertices with tag $X$ from left to right along
the backbone $P_1$. Since a solution $F$ contains at most $k$ free edges and the slices of 
the vertices of the backbone are disjoint, at most $k$ slices of vertices $w_i$ contain an edge of $F$. 
Hence, if we partition the set of all slices $SL(w_i)$ into 
$k+1$ classes according to their index $i$ modulo $k+1$, the solution $F$ will not
intersect one of these classes. We branch on these $k+1$ choices. Assume 
for instance that $F$ does not contain an edge in all $SL(w_i)$ where $i$ divides $k+1$. 
Therefore, we can safely contract each of such slices $SL(w_i)$ onto $w_i$. This makes $w_i$
a full vertex.

Let us now enumerate the full vertices $z_1,z_2,\dots $ from left to right. Let $uv$ be some
stable $XX$-edge. There exists a full vertex $z_i$ to the left of $u$ (with 
possibly $z_i=u$) and a full vertex $z_{i+1}$ to the right of $v$. Since 
the number of vertices with tag $X$ between 
$z_i$ and $z_{i+1}$ is at most $k$, the number of $XX$-edges between 
$z_i$ and $z_{i+1}$ is at most $k+1$. The \emph{rank} of $uv$ is the index of $uv$ in 
the enumeration of the edges between $z_i$ and $z_{i+1}$ from left to right. Every edge 
of $P_1$ has some rank between 1 and $k+1$. In particular, we can branch over 
the rank of the selected stable $XX$-edge. Assume for instance that the rank of our 
chosen edge is 1. We then contract all edges which are not stable $XX$-edges with rank 1. 
This leaves only full vertices on $P_1$ since by construction there is a 
full vertex between two edges of the same rank.
\end{proof}

Note that after performing the reduction of Lemma~\ref{seqlemons}, if $v_iv_{i+1}$ is a stable $XX$-edge, then for every vertex $w\in P_j$, with $j\in X$, which lies between $v_i$ and $v_{i+1}$ in $P_j$, every $wY$-path contains $v_i$ or $v_{i+1}$. In particular, if $X=\{1,\dots,\lambda\}$ then $v_i$ and $v_{i+1}$ are $xy$-cut vertices in $G_1$.

\begin{lemma}\label{onlylemons}
We can assume that $X=\{1,\dots ,\lambda\}$.
\end{lemma}

\begin{proof}
In other words, we can reduce to the case where every vertex of $P_1$ is a 
cutvertex of $G_1$. Assume now that $X$ is not equal to $\{1,\dots ,\lambda\}$. 
We show that we can partition our component $G_1$ into
two components $G_1^1$ and $G_1^2$. This partition leaves the free-connectivity
unchanged, but decreases the slice connectivity. 
A vertex $v_i$ of the backbone $P_1$ is \emph{left clean} if the edge $v_{i-1}v_i$
of $P_1$ is a stable $XX$-edge, but the edge $v_{i}v_{i+1}$ of $P_1$ is not. It is 
\emph{right clean} if the edge $v_{i}v_{i+1}$
is a stable $XX$-edge, but the edge $v_{i-1}v_i$ is not. Finally, $v_i$ is \emph{clean} 
if both $v_{i-1}v_i$ and $v_{i}v_{i+1}$ are stable $XX$-edges. 
When enumerating all left clean and right clean vertices from left to right,
we obtain the sequence of distinct vertices $r_1,l_1,r_2,l_2,\dots ,r_p,l_p$
where the $r_i$ are right clean and the $l_i$ are left clean vertices. 
Observe that $x$ and $y$ do not appear in the sequence since their tag 
is $\{1,\dots ,\lambda\}$.
Let us consider a pair $r_i,l_i$. We say that a vertex $v$ of 
$G_1$ is \emph{between} $r_i$ and $l_i$ if every path from $v$ to $x$
or $y$ intersects $\{r_i,l_i\}$. Let $B_i$ be the set of vertices which
are between $r_i$ and $l_i$. Let $B$ be the union of $B_i$
for $i=1,\dots ,p$. Let $G_1^1$ be a copy of the graph induced 
on $B$ by $G_1$. Observe that $G_1^1$ has $p$ connected components, since $l_i\neq r_{i+1}$.
We contract in $G_1^1$ the vertices $l_i$ and $r_{i+1}$, for all $i=1,\dots ,p-1$,
hence making $G_1^1$ connected. We finally identify in $G_1^1$ the vertex $x$ with $r_1$ and $y$ with $l_p$.
The backbone $P_1^1$ of $G_1^1$ simply consists of the edges of the original backbone.
To construct $G_1^2$, we remove from $G_1$ all the vertices of $B$
which are not left clean or right clean vertices. Hence no stable $XX$-edge is left
in $G_1^2$. We contract all the backbone edges of $G_1^2$. Formally, all the vertices between
$x$ and $r_1$ are contracted to a vertex $w_1:=x$, more generally 
all the vertices between $l_i$ and $r_{i+1}$ are contracted to a new vertex called $w_{i+1}$,
and finally all the vertices between $l_p$ and $y$ are contracted to $w_{p+1}:=y$.
We now add the path $w_1,w_2,\dots ,w_{p+1}$
which is the backbone $P_1^2$ of $G_1^2$.
We correlate the edges of the backbone of $G_1^1$ and $G_1^2$ by adding 
clauses implying that the chosen edge of $P_1^2$ is $w_iw_{i+1}$ if and only if  
the chosen edge of $P_1^1$ is between $r_i$ and $l_i$.
We finally branch to split the number of free edges $f_1$
chosen in $G_1$ into $f_1^1+f_1^2=f_1$, the respective free edges 
deleted in $G_1^1$ and $G_1^2$. Let us call $G'$ the graph $G$ in which 
$G_1$ is replaced by $G_1^1$ and $G_1^2$. Note that the free edges 
of $G_1$ are partitioned into the free edges of $G_1^1$ and of $G_1^2$.
Observe that the free-connectivity of $G$ and $G'$ are equal.
However, the slice connectivity has decreased in $G'$, since its value 
is $0$ in $G_1^1$ and strictly less than $\lambda -|X|$ in $G_1^2$. Indeed, for $i=1,\dots ,p-1$, the edge $l_il_{i+1}$ is either unstable or the tag of one of its endpoints strictly contains $X$. Hence contracting all vertices between $l_i$ and $r_{i+1}$ strictly increases the tag of the resulting vertex in $G_1^2$.
Hence we improve our invariant. Figure~\ref{L22} gives an example of this transformation.

\begin{figure}
\center\includegraphics[scale=.5,angle=-90]{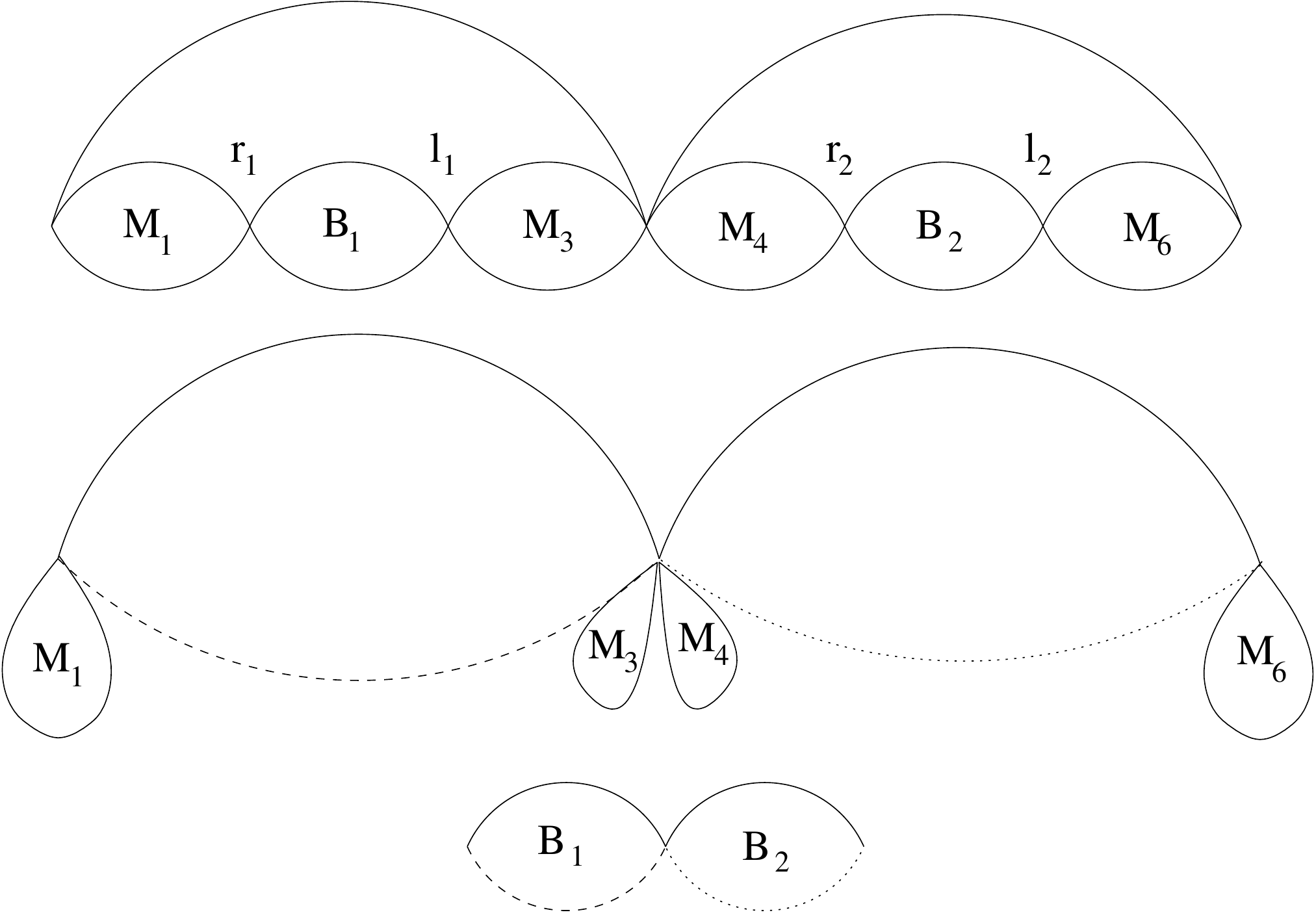}
\caption{At the top, a component $G_1$ before transformation in the proof of Lemma~\ref{onlylemons}. At the bottom, $G_1^2$ consists of $B_1$ and $B_2$, and $G_1^1$ is depicted in the center. The dashed edge in $G_1^1$ is correlated to the dashed edge in $G_1^2$, and similarly the dotted edges in $G_1^1$ and  $G_1^2$ are correlated.}
\label{L22}
\end{figure}

We now have to prove that there exists a multicut in $G'$ if and only if there exists
one in $G$ which uses a stable $XX$-edge. This comes from the following observation:
Let $e=v_jv_{j+1}$ be a stable $XX$-edge of $P_1$ between $r_i$ and $l_i$.
Let $G_e$ be obtained from $G$ by deleting $e$,
contracting $x$ to all vertices of $P_1$ to the left of $v_j$, and contracting $y$ to
all vertices of $P_1$ to the right of $v_{j+1}$. 
Let $G'_e$ be obtained from $G'$ by deleting $e$ in $P_1^1$, deleting the
edge $w_iw_{i+1}$ correlated to $e$ in $P_1^2$,
contracting $x$ to all vertices of $P_1^1$ to the left of $v_j$ and all 
vertices of $P_1^2$ to the left of $w_i$, and contracting $y$ to
all vertices of $P_1^1$ to the right of $v_{j+1}$ and all 
vertices of $P_1^2$ to the right of $w_{i+1}$.
The key fact is that $G_e$ is equal to $G'_e$. Hence the multicuts in $G$
and $G'$ selecting the edge $e$ are in one to one correspondence.
\end{proof}

The proof of Lemma~\ref{onlylemons} produces a new component, hence a new
edge to be chosen in a backbone. This increases the deletion allowance
by $1$, but the number of free edges has not increased. Since our invariant improves, we can perform this operation a bounded number of times, and this is indeed enough for 
our proof. More accurately, this operation can be performed at most $2k$ times. 
Indeed, when $X\neq \{1\}$, we have $f_1^1>0$ and $f_1^2>0$, hence 
the integer partition $f_1,\dots ,f_q$ has been refined, and this can 
happen at most $k$ times. When $X =\{1\}$, we can have $f_1^1=0$,
but in this case, the slice connectivity in 
$G_1^2$ is stricly less than $\lambda-1$, hence applying Lemma~\ref{onlylemons}
on $G_1^2$ results in a refinement of $f_1^2,f_2,\dots ,f_q$. This 
can be done at most $k$ times.

%%%%%%%%%%%%%%%%%%%%%%%%%%%%%%%%%%%%%%%%%%%%%%%%%%%%%%%%%%%%%%
%%%%%%%%%%%%%%%%%%%%%%%%%%%%%%%%%%%%%%%%%%%%%%%%%%%%%%%%%%%%%%
%%%%%%%%%%%%%%%%%%%%%%%%%%%%%%%%%%%%%%%%%%%%%%%%%%%%%%%%%%%%%%
%%%%%%%%%%%%%%%%%%%%%%%%%%%%%%%%%%%%%%%%%%%%%%%%%%%%%%%%%%%%%%

\subsection{Reducing the lemons}\label{dilworth}
We now assume that each vertex of the backbone $P_1$ of $G_1$
intersects all other paths $P_i$. 
Let $v_iv_{i+1}$ be an edge of the backbone $P_1$. The \emph{$v_i$-cherry} $C_i$
is the set of all vertices $u$ of $G_1$ such that every $uY$-path 
contains $v_i$.

\begin{figure}
\center\includegraphics[scale=.75,angle=-90]{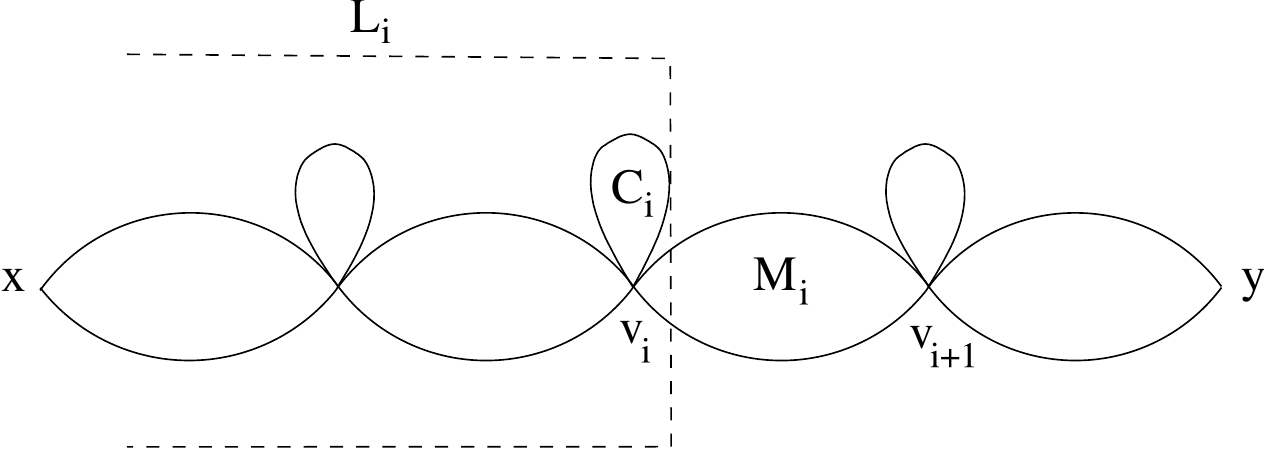}
\caption{The left subgraph $L_i$ of $v_i$. The set $C_i$ is the $v_i$-cherry, and $M_i$ is the lemon of the backbone edge $v_iv_{i+1}$.}
\label{kiki}
\end{figure}

The \emph{lemon} $M_i$ of $v_iv_{i+1}$ is the set consisting of $v_i$, $v_{i+1}$ and of all vertices $u$ of $G_1$ which do not belong to a cherry and such that every $ux$-path in $G_1$ contains $v_i$ and every $uy$-path in $G_1$ contains $v_{i+1}$.
Observe that when contracting $v_iv_{i+1}$, 
the lemon $M_i$ becomes part of the $v_i$-cherry, where $v_i$ denotes the 
resulting vertex. We denote by $L_i$ the union of all $C_j$ with
$j\leq i$ and all $M_j$ with $j<i$. We call $L_i$ the \emph{left subgraph}
of $v_i$. Similarly, the \emph{right subgraph} $R_i$
of $v_i$ is the union of all $C_j$ with
$j\geq i$ and all $M_j$ with $j>i$.  See Figure~\ref{kiki}.

If a multicut $F$ selects the edge $v_iv_{i+1}$ in the backbone, then the vertices $x,v_1,\dots,v_i$ all lie in the same connected component of $G\setminus F$. When these vertices $x,v_1,\dots,v_i$ are contracted to $x$, the set $L_i$ becomes an $x$-cherry. Half-requests through $y$ with an endpoint in $L_i$ are automatically cut since $F$ splits $Y$. Consider the terminals $T_i$ of half-requests of $L_i$ which are routed via $x$. By Theorem~\ref{cherryrequestreduction}
we can reduce $T_i$ to a bounded set of terminals $K_i$ (note that these half-requests become equivalent to usual requests, since $L_i$ is now an $x$-cherry). This motivates the following key definition.

By Lemma~\ref{finitelycut}, we define ${\mathcal L}_i$ to be a bounded active set of edges in the $x$-cherry obtained from $L_i$ by contracting vertices $x,v_1,\dots,v_i$. By Theorem~\ref{boundedcut}, we can compute such sets ${\mathcal L}_i$ so that ${\mathcal L}_j\cap L_i\subseteq  {\mathcal L}_i$ when $i\leq j$.

Let us say that a multicut $F$ selecting $v_iv_{i+1}$ in $P_1$ is \emph{proper} if
$F\cap L_i$ is included in  ${\mathcal L}_i$.

\begin{lemma}\label{proper}
If there exists a multicut $F$ of size at most $k$ containing the backbone edge $v_iv_{i+1}$, then there is a proper multicut $F'$ of size at most $k$ containing $v_iv_{i+1}$.
\end{lemma}

\begin{proof} 
Consider a multicut $F$ containing $v_iv_{i+1}$. As the set $\mathcal{L}_i$ is active in the cherry obtained by contracting the path $x,v_1,\dots, v_i$ in $L_i$, there exists a multicut $F'$ of size $k$ such that $F'\setminus L_i = F\setminus L_i$ and $F'\cap L_i \subseteq {\mathcal L}_i$. Hence $F'$ is proper and contains $v_iv_{i+1}$.
\end{proof}

We denote by $\mathcal L$ the set of all subsets $F$ of size at most $k$
contained in some $ {\mathcal L}_i$. We denote by $c$ the maximum size of a set ${\mathcal L}_i$. Note that $c$ is bounded in terms of $k$.

For two sets $F_i\subseteq {\mathcal L}_i$ and $F_j\subseteq {\mathcal L}_j$ with $j\ge i$, let us write $F_i\preceq F_j$ when $F_j\cap  L_{i+1}\subseteq F_i$. Observe that $\preceq$ is a partial order. A subset ${\mathcal F}$ of $\mathcal L$ is \emph{correlated}
if:
\begin{itemize}
\item elements of ${\mathcal F}$ have the same size, and
\item ${\mathcal F}$ is a chain for $\preceq$, \ie{}if for every $F_i$ and $F_j$ in ${\mathcal F}$, with $F_i\subseteq {\mathcal L}_i$, $F_j\subseteq {\mathcal L}_j$ and $j\ge i$, we have $F_j\cap  L_{i+1}\subseteq F_i$. 
\end{itemize}

\begin{lemma}\label{Lpartition}
There is a partition ${\mathcal F}_1, {\mathcal F}_2,\dots ,{\mathcal F}_{k(2c)^k}$ of  $\mathcal L$ into $k(2c)^k$
correlated sets.
\end{lemma}

\begin{proof}
Let us prove by induction on $\ell=0,\dots,k$ that there exists no antichain for $\preceq$ in $\mathcal L$ consisting of $(2c)^{\ell}+1$ sets of size at most $l$. This clearly holds for $\ell=0$. Assume that this holds for $\ell-1$. By contradiction, let $A=\{F_1,F_2,\dots ,F_{(2c)^{\ell} +1}\}$ be an antichain of sets of size at most $\ell$. Let $t_i$ be an integer such that $F_i\subseteq{\mathcal L}_{t_i}$ for $i=1,\dots,(2c)^{\ell} +1$. We assume that the sets $F_i$ are enumerated in such a way that $t_i\leq t_j$ whenever $i\leq j$. The set $F_1$ is incomparable to all sets $F_i$ with $i>1$, hence
$F_i\cap L_{t_1+1} \nsubseteq F_1$ for all $i>1$. In particular it 
is nonempty, hence all sets $F_i$, for $i=1,\dots,(2c)^{\ell} +1$, have an edge in $L_{t_1+1}$. The sets $F_i$ such that $t_i=t_1$ have an edge in ${\mathcal L}_{t_1}$ by definition. The sets $F_i$ such that $t_i>t_1$ have an edge in  ${\mathcal L}_{t_1+1}$ as ${\mathcal L}_{t_i}\cap L_{t_1+1}\subseteq {\mathcal L}_{t_1+1}$, by definition of the sets ${\mathcal L}_{i}$. Since the size of ${\mathcal L}_{t_1}\cup {\mathcal L}_{t_1+1}$ is at most $2c$, there exists a subset $B$ of $A$ of size at least $(2c)^{\ell-1}+1$ of sets $F_i$ sharing a same edge $e\in {\mathcal L}_{t_1}\cup {\mathcal L}_{t_1+1}$. The set $\{F\setminus e | F\in B\}$ has size $|B|\ge (2c)^{\ell-1}+1$ and is an antichain of sets of size at most $\ell-1$ by definition of $\preceq$. This contradicts the induction hypothesis.

By Dilworth's Theorem, there exists a partition of $\mathcal L$ into $(2c)^k$ sets totally ordered by $\preceq$, which can be be refined according to the cardinality to obtain a partition into $k(2c)^k$ correlated sets. Such a partition can be found in FPT time.
\end{proof}

Let us now consider such a partition ${\mathcal F}_1, {\mathcal F}_2,\dots ,{\mathcal F}_{k(2c)^k}$ of $\mathcal L$ into correlated sets. Observe that by Lemma~\ref{proper} we can restrict our search to multicuts of the following type in $G_1$:

\begin{itemize}
\item A backbone edge $v_iv_{i+1}$.
\item Other edges in the lemon $M_i$, which separate $v_i$ from $v_{i+1}$ in $M_i$.
\item Edges in ${\mathcal L}_i$.
\item Edges in ${\mathcal R}_i$, which is defined analogously to ${\mathcal L}_i$, with the roles of vertices $x$ and $y$ reversed.
\end{itemize}

%Dessin ??

\begin{lemma}\label{onlycherries}
We can assume that there are no cherries $C_i$. Moreover, if a multicut of size at most $k$ exists, there exists one which contains only edges in one lemon $M_i$.
\end{lemma}

\begin{proof}
By Lemma~\ref{proper}, if there exists a multicut $F$ containing the backbone edge $v_tv_{t+1}$, then there is a proper multicut $F'$ containing $v_tv_{t+1}$. By definition $F'\cap L_t \subseteq {\mathcal L}$. 

We branch over the existence of a proper solution $F'$ such that $F'\cap L_t\in {\mathcal F}_j$ for $j=1,\dots,k(2c)^k$, where $t$ is the integer such that $v_tv_{t+1}\in F'$. Let us assume 
that we are in the branch where $F'\cap L_t\in {\mathcal F}_j$. A backbone edge $v_iv _{i+1}$ is in the \emph{support} of ${\mathcal F}_j$ if there is some $F_i\in {\mathcal F}_j$
such that $F_i\subseteq  {\mathcal L}_i$. When $v_iv _{i+1}$ is in the support we say that lemon $M_i$ is a \emph{support lemon}. In this case, there actually exists a unique set in ${\mathcal F}_j$, which we denote by $F_i$,
such that $F_i\subseteq  {\mathcal L}_i$, as ${\mathcal F}_j$ is totally ordered under $\preceq$. Let $\ell$ be the number of edges of elements of ${\mathcal F}_j$.

\begin{claim}\label{support}
For all $F_a \in \mathcal{F}_j$, if $M_i$ is a support lemon then $F_a\cap M_i=\emptyset$.
\end{claim}
\begin{proof}
As ${\mathcal L}$ contains no backbone edge by definition, it is enough to show that $u$ is not disconnected from
$v_i$ in $G_1\setminus F_a$. As $M_i$ is a support lemon, there exists a set ${F_i\in \mathcal F}_j$ such that $F_i\subseteq {\mathcal L}_i$. Consider a set $F_a\in {\mathcal F}_j$ with 
$F_a\subseteq {\mathcal L}_a$. If $a\leq i$, then $F_a\subseteq L_a\subseteq L_i$,
hence $F_a\cap M_i=\emptyset$. If $a \geq i$, then $F_a \cap L_{i+1} \subseteq F_i \subseteq L_i$ as ${\mathcal F}_j$ is correlated, hence $F_a\cap M_i=\emptyset$ holds as well. This completes the proof of Claim~\ref{support}.
\end{proof}

Consider now a vertex $u$ such that either $u$ belongs to some cherry $C_i$ or $u$ belongs to a lemon $M_i$ which is not a support lemon. An edge $v_av_{a+1}$  
in the support \emph{affects} a half-request $(u,x,v)$ if $a<i$ or if $i\le a$ and the unique set $F_a\in {\mathcal F}_j$ such that $F_a\subseteq {\mathcal L}_a$ separates $u$ from $x$ in $G_1$. 
If $v_av_{a+1}$ does not affect $(u,x,v)$, then neither does $v_bv_{b+1}$ when $b\ge a$. Indeed when $b\ge a$, $F_b\subseteq {\mathcal L}_b$ and $F_b\in {\mathcal F}_j$, we have that $F_b\cap L_a\subseteq F_a$.%, hence $F_b$ does not separate $u$ from $x$ in $G_1$. 

Let us now modify the instance. If no edge of the support affects a half requests $(u,x,v)$, where $u$ belongs to some cherry $C_i$ or $u$ belongs to a lemon $M_i$ which is not a support lemon, we remove $(u,x,v)$ from $R$ and add the half-request $(x,x,v)$. Otherwise we let $v_av_{a+1}$ be the support edge with $a$ maximal which
affects $(u,x,v)$. We replace $(u,x,v)$ in $R$ by $(v_{a+1},x,v)$. We call this process \emph{projecting} the half-request $(u,x,v)$. After projecting all half-requests via $x$ with an endpoint in a cherry or in a lemon $M_i$ which is not a support lemon, we decrease $f_i$ by $\ell$
and contract every edge of $P_1$ which is not in the support of ${\mathcal F}_j$.
Note that if $v_iv_{i+1}$ is not in the support then there remains no half-request via $x$ in $M_i$.

Assume that $F'$ is a solution in this reduced instance which uses an edge
$v_av_{a+1}$ in the support. Let $F_a$ be the element of ${\mathcal F}_j$ such that $F_a\subseteq {\mathcal L}_a$. We have that $F'\cup F_a$ is a solution in the original instance. Indeed the requests in the support lemons are cut in $F' \cup F_a$ if and only if they are cut by $F'$ since $F_a$ does not intersect these lemons by Claim~\ref{support}. Also, the requests in the lemons which are not support lemons and in the cherries are cut in the reduced instance if and only if they are cut by $F_a$ in the initial instance by construction. % Expanser ??

Conversely, assume that $F$ is a proper solution in the original instance which
uses the edge
$v_av_{a+1}$ and such that  $F\cap L_a\in {\mathcal F}_j$. In particular $F_a= F\cap L_a$, hence $F\setminus F_a$ is a solution of the reduced instance. Indeed, all half-requests $(u,x,v)$ cut by $F_a$ in the original instance are affected by $v_av_{a+1}$, hence they have been projected to $(v_{i},x,v)$ with $i\ge a+1$, hence they are cut by $F\setminus F_a$ in the reduced instance.

The reduction, consisting in projecting all half-requests with an endpoint in a cherry or in a lemon which is not a support lemon, improves our invariant unless $\ell=0$, \ie{}unless the proper solution of the original instance
with backbone edge $v_iv_{i+1}$ does not use any edge in $L_i$. In this 
case, all the requests via $x$ of cherry $C_j$ are projected to $v_j$,
for all $j$. By the same argument, we can assume that no edge in a proper solution
is selected to the right of $M_i$ and that the half-requests via $y$ of $C_j$ 
are projected to $v_j$. 
Hence the cherries do not send any request, so we simply contract them.
We are only left with lemons, and we moreover know that if a solution
exists, then there exists one which uses only edges in a single lemon. This concludes the proof of Lemma~\ref{onlycherries}.
\end{proof} 

\begin{theorem}\label{newbackbone}
We can assume that $G_1$ only consists of the backbone $P_1$.
\end{theorem}

\begin{proof}
We assume that $\lambda>1$ and show that we can improve our invariant.
Let us consider a backbone edge $v_iv_{i+1}$. We denote by $W$ the multiset of vertices $\{w_2,\dots ,w_{\lambda}\}$
where $w_j$ is the vertex of the slice $S_i$ of $v_i$ in $M_i$ which belongs to the
path $P_j$ and has a neighbor in $M_i\setminus S_i$.
In other words, $w_j$ is the rightmost vertex of each path $P_j$ in the slice of
$v_i$. These vertices $w_j$ are not necessarily distinct, for instance if $v_i$
has degree $\lambda$ in $M_i$, the slice $S_i$ is exactly $\{v_i\}$ hence
all $w_j$ are equal to $v_i$ for $j=2,\dots,\lambda$. We also denote by $Z=\{z_2,\dots ,z_{\lambda}\}$
the multiset of vertices of the slice $T_i$ of $v_{i+1}$ in $M_i$ which belong respectively to the
paths $P_2,\dots ,P_{\lambda}$ and have a neighbor in $M_i\setminus T_i$.

A multicut $F$ induces a partition of $W\cup Z$ according to the components of 
$G\setminus F$. A vertex of $W\cup Z$ has three possible \emph{types}: 
it can be in the same component as $x$ after the removal of $F$, in the same component
as $y$, or in another component. Observe that if two vertices $a,b$ of $W\cup Z$
belong to components distinct from the components of $x$ and $y$ in $G\setminus F$, then 
$F$ is still a multicut after contracting $a$ and $b$. Hence $F$ induces
a partition of $W$ into three parts which can be contracted, still
leaving $F$ as a multicut. We now branch over all partitions of 
$W\cup Z$ into three parts $WZ_x,WZ_y,WZ_u$, where $WZ_x$ are vertices which
are in the same component as $x$, $WZ_y$ are vertices which
are in the same component as $y$, and $WZ_u$ are vertices of the same type, possibly disconnected from $x$ and $y$ (but not necessarily so). We branch over 
all possible partitions of $W$ into $WZ_x,WZ_y,WZ_u$, and 
contract in each branch $WZ_x$ to $v_i$, $WZ_y$ to $v_{i+1}$, and $WZ_u$ (if not empty) is 
contracted to a single vertex called $u_i$. These contractions are performed
simultaneously in all $M_i$. We denote by $G_1'$ the resulting component, by $M'_i$ the 
contracted lemon $M_i$, and by $S'_i$ the contracted $S_i$. 

If some vertex of $W$ belongs to $WZ_y$, or if some vertex of $Z$ 
belongs to $WZ_x$, or if $WZ_u$ intersects both $W$ and $Z$, then the $xy$ 
edge-connectivity increases in $G'_1$ since there exists an $xy$-path in $G'_1$ without edges of $\lambda(x,y)$-cut in $G_1$. Hence, we improve our invariant, but we cannot directly
conclude since the edges of the backbone may not be critical
anymore. Indeed, it can happen that $M'_i$ has connectivity 
(between $v_i$ and $v_{i+1}$) less than the connectivity of another lemon $M'_j$ in which case 
the backbone edge $v_jv_{j+1}$ is not critical. To get a correct instance of \BM, we simply branch on the connectivity 
of the lemon $M_i'$ corresponding to the chosen edge $v_iv_{i+1}$.
In the branch corresponding to connectivity $l$, we contract the 
backbone edges $v_iv_{i+1}$ where $M_i'$ has connectivity distinct from $l$.

Hence we can assume without loss of generality that $W$ is partitioned into $WZ_u$ and $WZ_x$, and that $Z=WZ_y$.
Since we contract $WZ_y$ to $v_{i+1}$, we obtain that $v_{i+1}$
has now degree $\lambda$ in $M'_i$, and $T_i$ is a $v_{i+1}$-cherry. Let us assume that $WZ_u\neq \emptyset$.
Since we have contracted the vertices of $W$ to $v_i$ and $u_i$,
the set $S'_i$ has exactly two
vertices with a neighbour in $M'_i\setminus S'_i$, namely $v_i$ and $u_i$. 
Note that the degree of $v_i$ in $M'_i\setminus S'_i$ is exactly 
the number of vertices $w_j$ chosen in $WZ_x$ (with multiplicity since $WZ_x$ is a multiset). We denote it by $d$. It does
not depend on $i$ since we have chosen in every $M_i$ the same 
subset $WZ_x$ inside $\{w_2,\dots ,w_{\lambda}\}$.

Let $\lambda _S$ be the $v_iu_i$ edge-connectivity
in $S'_i$. If $\lambda _S> f_1$, we simply contract
$v_i$ and $u_i$ since one cannot separate $v_i$ from $u_i$. 
We branch in order to assume that $\lambda _S$ is some fixed
value. In the branch corresponding to connectivity $\lambda_S$, we contract backbone edges $v_iv_{i+1}$ where $S_i$ has connectivity distinct from $\lambda_S$. Let $P'_1,\dots ,P'_{\lambda_S}$ be a
collection of edge disjoint paths from $v_i$ to $u_i$ in $S'_i$. 
We denote by $S'$ the slice of $v_i$ in $S'_i$, and again
consider the rightmost vertices $W'=\{w'_1,\dots ,w'_{\lambda_S}\}$
of $S'$ in the paths $P'_j$. We branch again over all possible 
partitions of $W'$ into $W'_x,W'_y,W'_u$. Again if $W'_y$ is not empty,
we increase the connectivity between $x$ and $y$. Observe that $W'_u$ can be contracted to $WZ_u$,
hence to $u_i$. In particular if $W'_u$ is not empty,
we increase the connectivity between $v_i$ and $u_i$ in $S_i'$. We iterate this 
process in $S'_i$ until either $W'_u$ is empty in which case $v_i$ has degree
$\lambda_S$ in $S'_i$, or $\lambda _S$ exceeds $f_1$ 
in which case we contract $v_i$ and $u_i$.

We apply Lemma~\ref{onlycherries} on $G'_1$. Therefore, we
can assume that no cherries are left and that 
if a solution exists, one multicut is contained in some $M'_i$.
Two cases can happen:

If $f_1\geq d+\lambda_S+\lambda-1$, and $v_iv_{i+1}$ is chosen, then we can assume that the restriction of the 
multicut to $M'_i$ simply consists of all the edges incident to $v_i$ and $v_{i+1}$ 
in $M'_i$. Indeed $v_i$ is incident to $d+\lambda_S$ free edges, and $v_{i+1}$ is incident to $\lambda -1$ free edges. This is clearly the best solution since it separates all vertices
of $M'_i\setminus\{v_i,v_{i+1}\}$ from $v_i$ and $v_{i+1}$. Therefore, we project every request 
$(u,x,v)$ where $u\in M'_i$ to $(v_{i+1},x,v)$ and project every request 
$(u,y,v)$ where $u\in M'_i$ to $(v_{i},y,v)$. Finally we reduce $f_1$ to 
$0$ and we delete all vertices of $G'_1$ which are not in $P_1$. 

Assume now that $f_1< d+\lambda_S+\lambda-1$. We branch over $2(\lambda-1)$ choices, 
where the branches are named $B_j$ and $B'_j$ for all $j=2,\dots ,\lambda$.
In the branch $B_j$, we assume that only one
edge of our solution is selected in $P_j$, and this edge is critical. In the branch $B'_j$,
we assume that all the edges of our solution selected in $P_j$ are not critical. Observe that
in $B'_j$, we contract non critical edges of $P_j$ and improve our invariant.
In the branch $B_j$, we find a new backbone $P_j$. In this last case, we delete the edges 
of $P_1$ and reduce the number
of free edges to $f_1-1$. We also translate the clauses in terms of edges of the new backbone $P_j$. Indeed the number of edges in the backbone of $G_1$ has changed. Clauses of the form $c_1\le i$ become $c_1\le \epsilon(i)$ where $\epsilon(i)$ denotes the index of the rightmost edge of $P_j$ in the lemon $M'_i$. 

This branching process covers all the cases 
where $v_i=u_i$ since in this case $f_1<2\lambda-2$ and therefore one path $P_j$
contains only one edge of the multicut. In the case $v_i\neq u_i$, assume that a multicut $F$ is not of a type treated in one of our branches. In other words, $F$ contains at least two edges in each path $P_j$ for $j=2,\dots,\lambda$, and at least one of them is critical. Then $F$ contains two edges in each of 
the $d$ paths $P_j$ not containing $u_i$ since $F$ does not respect the branches $B_j$ for $j=2,\dots,\lambda$. Also, $F$ contains one edge outside $S'_i$ in each path $P_j$ containing $u_i$ since edges in $S'_i$ are not critical and $F$ is not treated in the branches $B'_j$. Thus $F$ contains
at least $2d+(\lambda-d-1)$ free edges outside
$S'_i$. Hence less than $\lambda _S$ edges of $F$ lie in $S'_i$, thus $v_i$ and $u_i$ belong to the same component in $G-F$. This case is covered in another branch in which $v_i$ and $u_i$ are contracted. Hence this branching process is exhaustive, and this completes the proof of Theorem~\ref{newbackbone}.
\end{proof}

%%%%%%%%%%%%%%%%%%%%%%%%%%%%%%%%%%%%%%%%%%%%%%%%%%%%%%%%%%%%%%
%%%%%%%%%%%%%%%%%%%%%%%%%%%%%%%%%%%%%%%%%%%%%%%%%%%%%%%%%%%%%%
%%%%%%%%%%%%%%%%%%%%%%%%%%%%%%%%%%%%%%%%%%%%%%%%%%%%%%%%%%%%%%
%%%%%%%%%%%%%%%%%%%%%%%%%%%%%%%%%%%%%%%%%%%%%%%%%%%%%%%%%%%%%%
%%%%%%%%%%%%%%%%%%%%%%%%%%%%%%%%%%%%%%%%%%%%%%%%%%%%%%%%%%%%%%
%%%%%%%%%%%%%%%%%%%%%%%%%%%%%%%%%%%%%%%%%%%%%%%%%%%%%%%%%%%%%%
%%%%%%%%%%%%%%%%%%%%%%%%%%%%%%%%%%%%%%%%%%%%%%%%%%%%%%%%%%%%%%
%%%%%%%%%%%%%%%%%%%%%%%%%%%%%%%%%%%%%%%%%%%%%%%%%%%%%%%%%%%%%%
%%%%%%%%%%%%%%%%%%%%%%%%%%%%%%%%%%%%%%%%%%%%%%%%%%%%%%%%%%%%%%
%%%%%%%%%%%%%%%%%%%%%%%%%%%%%%%%%%%%%%%%%%%%%%%%%%%%%%%%%%%%%%
%%%%%%%%%%%%%%%%%%%%%%%%%%%%%%%%%%%%%%%%%%%%%%%%%%%%%%%%%%%%%%
%%%%%%%%%%%%%%%%%%%%%%%%%%%%%%%%%%%%%%%%%%%%%%%%%%%%%%%%%%%%%%
%%%%%%%%%%%%%%%%%%%%%%%%%%%%%%%%%%%%%%%%%%%%%%%%%%%%%%%%%%%%%%
%%%%%%%%%%%%%%%%%%%%%%%%%%%%%%%%%%%%%%%%%%%%%%%%%%%%%%%%%%%%%%
%%%%%%%%%%%%%%%%%%%%%%%%%%%%%%%%%%%%%%%%%%%%%%%%%%%%%%%%%%%%%%
%%%%%%%%%%%%%%%%%%%%%%%%%%%%%%%%%%%%%%%%%%%%%%%%%%%%%%%%%%%%%%
\subsection{Reducing to 2-SAT}\label{subdividedmulticut}
We are left with instances in which the $Y$-components with two attachment vertices 
consist of backbones. We now reduce the last components.

\begin{lemma} \label{Ycherry}
We can assume that there is no component with one attachment vertex.
\end{lemma}

\begin{proof}
Let $Y=\{y_1,\dots ,y_p\}$ and $k$ be the number of free edges in our multicut. A vertex $y_i \in Y$ is \emph{safe} if for each pair of components attached only to $y_i$ there are no internal requests. If $y_i$ is not safe then there is a request $(u,y_i,v)$ in the union of the two components attached to $y_i$, hence $y_i$ must be disconnected from $u$ or from $v$ by the solution. We explore one branch where $u$ is added to $Y$, and one branch where $v$ is added to $Y$. This creates a component with two attachment vertices. This component has a backbone, and then the number of free edges decreases.

Hence we can assume that all the vertices of $Y$ are safe. The $y_i$-cherry is the union of all the components attached to $y_i$. We branch over all possible integer partitions of $k$ into a sum $k_1+k_2+\dots +k_p=k$. In each branch, we require that $k_i$ edges are deleted in the $y_i$-cherry for $i=1,\dots ,p$. By Lemma~\ref{finitelycut}, the $y_i$-cherry has a bounded active set $\mathcal{L}_i$, hence in the $y_i$-cherry we can consider only a bounded number of cuts of size $k_i$: all subsets of $\mathcal{L}_i$ of size $k_i$.
We then branch over these different choices. In a given branch, we delete a particular set of edges $F_i$ in the $y_i$-cherry. Thus, we delete the vertices of the $y_i$-cherry isolated from $y_i$ by $F_i$, and contract the other vertices of the $y_i$-cherry to $y_i$. Finally, no $Y$-cherry remains.
\end{proof}

\begin{theorem} \label{subdividedFPT}
Multicut is FPT.
\end{theorem}

\begin{proof}

By Lemma \ref{Ycherry}, we are only left with a graph $G$ which is a subdivision of 
a graph with at most $k$ edges, and a multicut must consist of exactly
one edge in each subdivided edge. 
Let us consider a half-request $(v_i,x,v'_j)$. Assume without loss of generality 
that $v_i\in G_1$, $v'_j\in G_2$, and $x$ belongs to $G_1$ and $G_2$ (if $x$ does not belong 
to $G_1$ or $G_2$, then splitting $Y$ automatically results in cutting the half-request $(v_i,x,v'_j)$).
For simplicity, we assume that the edges
of both $P_1$ and $P_2$ are enumerated in increasing 
order from $x$. We add to $\mathcal C$ the clauses $x_1\geq  i\Rightarrow x_2\leq j-1$
and $x_2\geq  j\Rightarrow x_1\leq i-1$. We transform all the half requests in this way.
Hence we are only left with a set of clauses which we have to satisfy.

We add all the relations $x_i\geq a\Rightarrow x_i\geq a-1$ and $x_i\leq a\Rightarrow x_i\leq a+1$ and $x_i\geq a\Rightarrow \lnot (x_i\leq a-1)$ and $x_i\leq a\Rightarrow \lnot (x_i\geq a+1)$.
We now have a 2-SAT instance which is equivalent to the original multicut instance. As 2-SAT is solvable in polynomial time, this shows that \BM{} is FPT. Hence the simpler \CM{} problem is FPT. Together with Theorem~\ref{multitomultiwithback} which reduces \M{} to \CM, this concludes the proof of Theorem~\ref{subdividedFPT}.
\end{proof}

%%%%%%%%%%%%%%%%%%%%%%%%%%%%%%%%%%%%%%%%%%%%%%%%%%
%%%%%%%%%%%%%%%%%%%%%%%%%%%%%%%%%%%%%%%%%%%%%%%%%%%
%%%%%%%%%%%%%%%%%%%%%%%%%%%%%%%%%%%%%%%%%%%%%%%%%%%
%%%%%%%%%%%%%%%%%%%%%%%%%%%%%%%%%%%%%%%%%%%%%%%%%%%%%
%%%%%%%%%%%%%%%%%%%%%%%%%%%%%%%%%%%%%%%%%%%%%%%%%%%%%%%%
%%%%%%%%%%%%%%%%%%%%%%%%%%%%%%%%%%%%%%%%%%%%%%%%%%%%%
%%%%%%%%%%%%%%%%%%%%%%%%%%%%%%%%%%%%%%%%%%%%%%%%%%%%%

\section{Improving the running time to single exponential}\label{time}

The main problem to get a single exponential bound for our FPT
algorithm is Theorem~\ref{extractleft} which uses $\Delta$-systems.
Let us improve this bound by the following result.

\begin{theorem}\label{ERnew}
Every set $K$ with at least $k'^{{k+2\choose 2}-1}$ vertices of $G$ contains a subset 
$K'$ of size $k'$ such that every left cut $S$ with
$\delta(S)\leq k$ satisfies either $S\cap K'=\emptyset$ or $|K'\setminus S|\leq k$.
The set $K'$ can be computed in FPT single exponential time.
\end{theorem}

\begin{proof}
Observe that the result trivially holds when $k'\leq k$. So we
can assume $k'>k$.
Recall that $G$ is rooted at $x$. We prove the result by induction 
on $k$. This is clear for $k=1$ since the complement of left cuts forms a collection of disjoint sets, hence inducing a partition
of $K$. We have either a class $K'$ of this partition containing at least 
$\sqrt{|K|}$ elements, or one can find a set $K'$ of size at least 
$\sqrt{|K|}$ which elements are chosen in different classes. In both cases 
$K'$ satisfies the induction hypothesis. Assuming $k>1$, two cases can happen:

\begin{itemize}
\item There exists a left indivisible cut $S$ with $\delta (S)\leq k$ and 
$|K\setminus S|\geq k'^{{k+1\choose 2}-1}$. By induction, we extract from $K\setminus S$ 
a subset $K'$ with size $k'$ such that every left cut $S$ with
$\delta(S)\leq k-1$ satisfies either $S\cap K'=\emptyset$ or $|K'\setminus S|\leq k-1$.
To conclude, let us consider a left cut $S'$ with $\delta (S')=k$. If $S'=S$,
$S'$ isolates $K'$, hence we assume that $S'$ is distinct from $S$. Observe
that $K'\setminus S'$ is equal to $K'\setminus (S\cup S')$. Since $S\cup S'$
is a left cut with border at most $k-1$, we then have $K'\setminus S'$ is either
$K'$ or has size at most $k-1$. Hence our conclusion holds.

\item Assume that all left indivisible cuts $S$ with $\delta (S)\leq k$ satisfy 
$|K\setminus S|<k'^{{k+1\choose 2}-1}$. Let us form a graph $H$ with vertex set $K$ and
where $vv'$ is an edge when there exists a left indivisible cut $S$ 
with $\delta (S)\leq k$ such that $\{v,v'\}\cap S=\emptyset$. The degree
of a vertex $v$ of $H$ is less than $d:=k'^{{k+1\choose 2}-1}.k!$ since the number of 
left indivisible $xv$-cuts with border at most $k$ is at most $k!$. Note
that $d$ is less than $k'^{{k+1\choose 2}-1}.k'^k$ since $k'\geq k$. So there
is a stable set $K'$ in $H$ of size at least $|K|/d$, i.e. at least $k'$ since ${k+2\choose 2}-{k+1\choose 2}-k=1$.
Observe that every indivisible left 
cut $S$ with $\delta (S)\leq k$ isolates at most one vertex of $K'$. Hence 
every left  
cut $S$ with $\delta (S)\leq k$ isolates at most $k$ vertices of $K'$.
\end{itemize}
\end{proof}

The function $h$ of Theorem~\ref{extractright} is at most $\ell ^{O(\ell)}$, so the 
function $b$ in Theorem~\ref{cherryrequestreduction} is at most $k^{O(k)}.(k^{O(k)})^{k^2}=k^{O(k^3)}$.
So each node of our branching algorithm is computed in FPT time
with a single exponent. The depth being polynomial, the overall 
complexity has a single exponent.

\section{Vertex Multicut is FPT}\label{vertex}

We propose here a sketch of a translation of our proof for edge-multicut
in terms of vertex-multicut.
The proof has the same outline, hence we just explain how the notions
introduced for
edge-multicut can be transferred to the vertex-multicut setting. We
prove that the
following version of Multicut is FPT.

\begin{quote}
\noindent
\VM{}:\\
\textbf{Input}: A graph $G$, a set of requests $R$, a subset of vertices
$S$, an integer $k$.\\
\textbf{Parameter}: $k$.\\
\textbf{Output}: TRUE if there is a vertex-multicut of size at most $k$
which does not intersect $S$, otherwise FALSE.
\end{quote}

Observe that this is equivalent to the standard version of \VM{} when the set $S$
is empty. Let us now explain how we can translate the results of the
previous sections
for \VM{}.

For Section \ref{sectionconnectivity}, the results are based on the
submodularity of edge cuts. The vertex cuts being also submodular, we can
transfer the
results for vertices. Here an indivisible $xy$-cut is a
set of
vertices $K$ which deletion separates $x$ from $y$ such that no strict
subset of $K$
separates $x$ from $y$. For the reduction from \VM{} to \CM{}, the proof
is essentially the
same. One particularity of \VM{} is the following. When we contract
vertices, we
have to add the resulting vertex to $S$, the set of non deletable
vertices. Let $Y$
be the vertex-multicut of size $k+1$ given by iterative compression. We can
branch to decide which vertices of $Y$ belong to the solution and
then branch over the possible contractions of the set $Y$. Hence we can
assume that $Y \subseteq S$. Notice that we have to replace ``we add a
vertex to $Y$'' by ``we branch to know if the vertex is added to $Y$ or if
it belongs to the solution''. The connectivity between $x$ and $y$ is
the maximum number of paths between $x$ and $y$ which are disjoint on the set of deletable
vertices. The connectivity can be calculated by flows with weight $1$ for deletable
vertices and $\infty$ for non-deletable vertices. A vertex of
$\lambda$-cut is a deletable vertex which deletion
decreases the connectivity. In the vertex-multicut context,
a \emph{backbone} is a path in which only one vertex is deleted and
where every odd vertex belongs to the set $S$. In addition all the
vertices of the backbone are vertices of $\lambda$-cut.

To prove the existence of a backbone, we  have to generalize Lemma
\ref{degreexy}. The border of the slice of $x_i$ has size at most $k$ but the number of
vertices which touch
this border can be arbitrarily large. We can branch to know if a vertex is
deleted
in the slice. If this is not the case then the slice can be contracted to
$x_i$, hence
$x_i$ has only $\lambda$ neighbours. Otherwise we can branch to know if
each vertex in
the border is in the component of $y_i$ or a new component. In each of
these case
the invariant improves. Hence the only case which is uncovered is the case
when all
the vertices are with $x$. We can contract $x$ with the border of its
slice and we
have a cherry in which we have to delete vertices. By Lemma
\ref{finitelycut}, we can
bound the number of possible cuts. We can branch over these cuts and
decrease the
deletion allowance.

\begin{figure}\center\includegraphics[scale=.75]{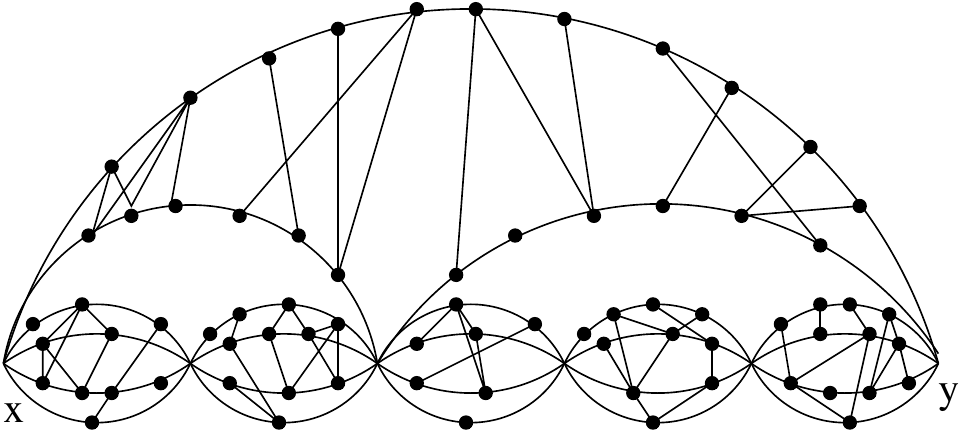}\caption{The
transformation into lemons for \VM. The backbone is the path at the
bottom.}\label{vertexlemons}\end{figure}

Let us now explain how we can prove that the \BM{} problem is FPT. A key notion
of Section~\ref{trlemon} is the notion of full vertex. We have to modify this
notion since contracted vertices are not deletable. Hence all the vertices
of the backbone cannot be full as for Edge-Multicut. Actually, we
transform the instance into an instance where the non-deletable vertices
are full (see Figure~\ref{vertexlemons}). The slice $S(v)$ of a
non-deletable vertex $v$ is the connected component of $v$ in $G$ minus the
vertices of $\lambda$-cut. We define the tag as for
edge-multicut. A vertex $v$ of the backbone is \emph{$X$-stable} if it can
be deleted, and the tags of its two neighbours are $X$ and the tag after
the contraction of $v$ with its two neighbors is still $X$. As for
edge-multicut, we can assume that we delete an $X$-stable vertex in the
backbone. We can similarly define classes for Lemma
\ref{seqlemons}, and remark that one class does not intersect the solution.
All the vertices of each slice in this class can be contracted. This ensures that we
can assume that all the vertices which are non deletable are full. We can write as in Lemma
\ref{onlylemons} that a non-deletable vertex is \emph{left} (resp. \emph{right})
\emph{clean} if the vertex to its left (resp. right) is $X$-stable and
then we can assume that $X=\{1,...,\lambda\}$ as for edge-multicut.

In the reduction
of the lemons for Vertex-Multicut, we cannot contract $x$ with the border of its slice since it does not ensure that the degree of $x$ is $\lambda$. Hence we have to contract $x$ with vertices of its slice which touch the vertices of the border. The set of such vertices can be restricted to a bounded size with Lemma~\ref{finitelycut}. Hence the same inductive method used for edge-multicut also holds.

\bibliographystyle{plain}
\bibliography{thebibliography}

\end{document}